\documentclass[11pt]{article}
\usepackage{times}
\usepackage[dvips, paper=letterpaper, top=1in, bottom=1in, left=1in, right=1in]{geometry}
\usepackage{fullpage}
\usepackage{amsmath,amssymb,amsthm,amsfonts,latexsym,bbm,xspace,graphicx,float,mathtools,dsfont,bm}
\usepackage{xcolor}
\usepackage{enumitem}
\usepackage[textsize=tiny]{todonotes}
\usepackage[numbers,sort,compress]{natbib}
\usepackage[normalem]{ulem}
\usepackage{bm}
\usepackage{thm-restate}
\usepackage{hyperref}
\usepackage{tikz}
\usepackage{cleveref}
\usepackage{array,multirow}
\usepackage{booktabs}

\newcommand{\bE}{\mathbb{E}}

\newcommand{\cE}{\mathcal{E}}
\newcommand{\cF}{\mathcal{F}}

\newcommand{\eps}{\varepsilon}
\newcommand{\prob}[1]{\Pr\left[ {#1} \right]}

\newcommand{\probg}[2]{\Pr\left[ {#1} \vert {#2} \right]}
\newcommand{\expect}[1]{\E\left[ {#1} \right]}
\newcommand{\expects}[2]{\E_{ #1 } \left[ {#2} \right]}
\newcommand{\expectg}[2]{\E\left[ {#1} \Big\vert {#2} \right]}

\newcommand{\ind}[1]{\mathds{1}\left[ {#1} \right]}
\newcommand{\iold}{i^{\mathsc{O}}}
\newcommand{\jold}{j^{\mathsc{O}}}
\newcommand{\inew}{i^{\mathsc{N}}}
\newcommand{\jnew}{j^{\mathsc{N}}}

\newcommand{\oldb}{b^{\mathsc{O}}}
\newcommand{\newb}{b^{\mathsc{N}}}
\newcommand{\olds}{s^{\mathsc{O}}}
\newcommand{\news}{s^{\mathsc{N}}}

\newcommand{\Bold}{B_{\mathsc{Old}}}
\newcommand{\Bnew}{B_{\mathsc{New}}}
\newcommand{\Sold}{S_{\mathsc{Old}}}
\newcommand{\Snew}{S_{\mathsc{New}}}

\newcommand{\BO}{\mathsc{BO}}
\newcommand{\BN}{\mathsc{BN}}
\newcommand{\SO}{\mathsc{SO}}
\newcommand{\SN}{\mathsc{SN}}

\newcommand{\qb}{q_{B}}
\newcommand{\qs}{q_{S}}

\newcommand{\Bb}{\mathbf{b}}

\newcommand{\Bq}{\textbf{q}}
\newcommand{\Bs}{\mathbf{s}}

\DeclareMathOperator*{\E}{\bE}
\DeclareMathOperator*{\Binom}{Binom}

\newtheorem{theorem}{Theorem}[section]
\newtheorem{lemma}[theorem]{Lemma}
\newtheorem{definition}[theorem]{Definition}
\newtheorem{corollary}[theorem]{Corollary}
\newtheorem{claim}[theorem]{Claim}

\newtheorem{fact}[theorem]{Fact}
\newtheorem{remark}[theorem]{Remark}

\newtheorem{proposition}[theorem]{Proposition}

\numberwithin{equation}{section}

\newcommand{\mathsc}[1]{{\normalfont\textsc{#1}}}
\newcommand{\OPT}{\mathsc{Opt}}
\newcommand{\TR}{\mathsc{Tr}}
\newcommand{\STR}{\mathsc{Str}}

\newcommand{\chris}[1]{\todo[inline,size=footnotesize,color=blue!20]{\textbf{Chris:} {#1}}}
\usepackage{etoolbox}
\newbool{draft}
\booltrue{draft}

\ifbool{draft}{
}{

}

\newcommand{\comment}[1]{}

\title{The Power of Two-sided Recruitment in Two-sided Markets}
 \author{
     Yang Cai\footnote{Supported by a Sloan Foundation Research Fellowship and the National Science Foundation Award CCF-1942583 (CAREER).} \\ Yale University \\ \texttt{yang.cai@yale.edu}
     \and
     Christopher Liaw \\ Google \\ \texttt{cvliaw@google.com}
     \and
     Aranyak Mehta \\ Google \\ \texttt{aranyak@google.com}
     \and
     Mingfei Zhao \\ Google \\ \texttt{mingfei@google.com}
 }
\allowdisplaybreaks

\begin{document}
\maketitle
\thispagestyle{empty}
\begin{abstract}
We consider the problem of maximizing the gains from trade (GFT) in two-sided markets. The seminal impossibility result by Myerson and Satterthwaite \cite{MyersonS83} shows that even for bilateral trade, there is no individually rational (IR), Bayesian incentive compatible (BIC) and budget balanced (BB) mechanism that can achieve the full GFT. Moreover, the optimal BIC, IR and BB mechanism that maximizes the GFT is known to be complex and heavily depends on the prior.

In this paper, we pursue a Bulow-Klemperer-style question, i.e.,~does augmentation allow for prior-independent mechanisms to compete against the optimal mechanism? Our first main result shows that in the double auction setting with $m$ i.i.d.~buyers and $n$ i.i.d.~sellers, by augmenting $O(1)$ buyers and sellers to the market, the GFT of a simple, dominant strategy incentive compatible (DSIC), and prior-independent mechanism in the augmented market is at least the optimal in the original market, when the buyers' distribution first-order stochastically dominates the sellers’ distribution. The mechanism we consider is a slight variant of the standard Trade Reduction mechanism due to McAfee \cite{McAfee92}. For comparison, Babaioff, Goldner, and Gonczarowski \cite{BabaioffGG20} showed that if one is restricted to augmenting only one side of the market, then $n(m + 4\sqrt{m})$ additional agents are sufficient for their mechanism to beat the original optimal and $\lfloor \log_2 m \rfloor$ additional agents are necessary for any prior-independent mechanism.

Next, we go beyond the i.i.d.~setting and study the power of two-sided recruitment in more general markets. Our second main result is that for any $\eps > 0$ and any set of $O(1/\eps)$ buyers and sellers where the buyers' value exceeds the sellers' value with constant probability, if we add these additional agents into any market with arbitrary correlations, the Trade Reduction mechanism obtains a $(1-\eps)$-approximation of the GFT of the augmented market. Importantly, the newly recruited agents are agnostic to the original market.
\end{abstract}
\clearpage
\setcounter{page}{1}
\section{Introduction}
In this paper, we study the problem of maximizing the gains from trade (GFT) in two-sided markets.
Two-sided markets are ubiquitous and have many practical applications;
some major examples include the FCC spectrum auction and online marketplaces such as Uber, Lyft, and Airbnb.
For example, in ride-sharing platforms, passengers (as the role of buyers) have some private value for obtaining transportation services and drivers (as the role of sellers) have some private cost for providing the necessary services.
Mechanism design for two-sided markets poses additional challenges over its one-sided counterpart.
In a one-sided market, the mechanism designer aims to maximize some objective (e.g.~welfare or revenue) subject to a one-sided incentive-compatibility constraint.
The seminal papers of Vickrey \cite{Vickrey61} and Myerson \cite{Myerson81} described how to design mechanisms that achieve the optimal welfare and revenue for one-sided markets, respectively.
However, in a two-sided market, one needs to ensure incentive compatibility for \emph{both} sides of the market as well as to ensure that the mechanism itself does not run a deficit (called \emph{budget balance}). The seminal impossibility result of Myerson and Satterthwaite \cite{MyersonS83} show that these additional constraints make the mechanism design problem much more challenging.
In particular, even in the simplest setting with  a single seller selling a single item to a single buyer (known as bilateral trade), no mechanism can achieve full efficiency while being Bayesian incentive-compatible (BIC), individually rational (IR), and budget balanced (BB). Myerson and Satterthwaite \cite{MyersonS83} also described the best BIC, IR and BB mechanism that maximizes efficiency in bilateral trade. However, the optimal mechanism is complex and heavily depends on the prior.

Motivated by the aforementioned challenges, there has been extensive research efforts and substantial progress in recent years investigating the two-sided market in the ``simple versus optimal'' perspective, i.e. to show that a simple mechanism can approximate the performance of the optimal mechanism.
A non-exhaustive list includes \cite{BabaioffCGZ18, CaiGMZ21, BrustleCWZ17, DengMSW22, KangPV22, MaoPW22, Fei22, CaiW23,BlumrosenD21,BlumrosenM16,colini2017fixed, Colini-Baldeschi16, colini2020approximately, DuttingRT14}. However, in many of these results, the mechanism designer requires a priori knowledge of both the buyers' and the sellers' distribution. Alternatively, the designer can increase the competition and thus her objective by recruiting more agents to the market.
In single-item one-sided markets, the seminal work by Bulow and Klemperer \cite{BulowK94} showed that the revenue of a second price auction with only a \emph{single} additional participant from the same population is at least that of the optimal mechanism with the original set of participants.
More recently, this result has been generalized to other one-sided market settings \cite{HartlineR09, EdenFFTW17, FeldmanFR18, FuLR19, BeyhagiW19, SivanS13, DughmiRS09, LiuP18}.
Such results showcase how additional competition, coupled with a simple mechanism
can be used to overcome the requirement of having precise knowledge of the underlying distributions and using the optimal-yet-complex mechanism. A natural question is whether such Bulow-Klemperer-type results also hold in the two-sided market settings.

In a recent paper, Babaioff, Goldner, and Gonczarowski \cite{BabaioffGG20} initiated this line of work to develop Bulow-Klemperer-type results for a fundamental single-parameter two-sided market setting called a \emph{double auction}.
In this problem there are $n$ sellers that each hold an identical item and the value of each seller is drawn i.i.d.~from some distribution $F_S$.
There are $m$ buyers that each wish to obtain one of these items and their value is drawn i.i.d.~from another distribution $F_B$.
They consider a variant of the natural, prior-independent Trade Reduction mechanism \cite{McAfee92}, which they call Buyer Trade Reduction (BTR).
They show that when the two distributions are identical,\footnote{Note that no further assumptions are placed on this distribution, while similar results in one-sided markets make certain regularity assumption about the distribution.}
%when both sides of the market have values drawn i.i.d.~from a common distribution
BTR with one additional buyer can achieve welfare at least the best welfare achievable in the original market even without the BIC, IR, and BB constraints (such a benchmark is called the \emph{first best}). Note that the first best is exactly the welfare of the celebrated VCG mechanism, which is BIC and IR, but may violate the BB constraint.
% \aranyaknote{(don't need to add, just curious) Does a simple mechanism give an approx to first-best (or second best) in the identical distributions case?}
%as efficient as the \emph{optimal allocation} (i.e.~even ignoring BIC, IR, and BB considerations) before augmentation.
%They then consider the setting where the two sides of the markets may have different distributions but where each side of the market remains i.i.d.
%With no assumptions, they prove that no finite bound is possible, even in the setting of bilateral trade.
While this resolves the most basic case, it is rarely assumed that both distributions are identical.
For example, a passenger in ride-sharing applications likely has higher value for obtaining transportation than the driver's cost for providing such transportation.
When there are no assumptions on the distributions, \cite{BabaioffGG20} prove that no finite bound is possible.
% Despite this negative result
When the buyers' distribution first-order stochastically dominates\footnote{A distribution $D$ first-order stochastically dominates $D'$ if $\Pr_{x\sim D}[x\leq c]\leq \Pr_{x\sim D'}[x\leq c]$ for every $c$.} the sellers' distribution, they prove that $n(m + 4\sqrt{m})$ additional buyers are sufficient for BTR to have welfare at least the first-best welfare in the original market
%optimal allocation where %$m$ is the number of buyers, $n$ is the number of sellers, and 
when $m \geq n$.\footnote{Their result applies to the $m\leq n$ case using \emph{Seller Trade Reduction}.}

An immediate question from Babaioff, Goldner, and Gonczarowski's result is whether the number of additional agents can be improved. Does a constant number of agents suffice for any number of buyers and sellers? The main difficulty turns out to come from the mechanism recruiting only one side of the market. In fact, their paper shows that $\lfloor \log_2 m \rfloor$ buyers are necessary if only extra buyers are recruited, even when there is a single seller. However in many situations, %it is reasonable to assume that one 
the mechanism designer is able to recruit both buyers and sellers. For example, in ride-sharing applications, recruiting both sides is very much feasible -- more riders will use the platform with better marketing, advertisement, or deals, and more drivers will adopt the platform with better incentives and marketing towards them.
In this paper, we allow recruiting from \emph{both} sides of the market. We show that with the Seller Trade Reduction (STR), a mechanism analogous to BTR, 
%which is very similar to BTR, with 
only $O(1)$ additional agents suffice.
%guarantees GFT at least the first-best GFT in the original market. See \Cref{thm:main_str} for the formal statement.
We give a formal definition of the mechanism in \Cref{subsec:results}.

In the above result, we assume that all the agents are independent, that all the buyers are drawn from a common buyer distribution, that all the sellers are drawn from a common seller distribution, and that the buyer distribution first-order stochastically dominates the seller distribution. These are the same assumptions that were made by \cite{BabaioffGG20}. We next turn to the setting where we make minimal assumptions on the market and ask about the power of two-sided augmentation in such a general setting.
Our second main result is that for any $\eps > 0$ and any set of $O(1/\eps)$ buyers and $O(1/\eps)$ sellers where the buyers' value exceeds the sellers' value with constant probability, the following holds.
If we augment these buyers and sellers into any market then Trade Reduction achieves a $(1-\eps)$-approximation of the optimal efficiency of the augmented market. We stress that the augmentation requires zero knowledge of the original market.
We also note that one-sided augmentation cannot be done in an agnostic manner. For example, suppose we augment the market with buyers that happens to have values less than all the sellers. It is not hard to see that for any prior-independent mechanism that is incentive-compatible, IR, and BB, its GFT remains unchanged after this augmentation.

To formally state our results, it is crucial to first discuss the measure of efficiency we adopt in this paper.  There are two main measures of efficiency in two-sided markets.
The first is the standard notion of \emph{welfare} in the literature, which is equal to the sum of the value of all buyers and sellers that hold the items in the final allocation.
The second is the \emph{gain from trade (GFT)} which is the welfare of the final allocation minus the total value of sellers.
At a high-level the GFT of a mechanism is a direct measure of the \emph{additional} value of a mechanism.
Note that when the set of sellers is fixed, maximizing welfare in a market where only the buyers are augmented (as in \cite{BabaioffGG20}) is identical to maximizing the GFT.
In other words, there is no need to make a distinction between welfare and GFT.
However, since we are interested in the problem where both sides of the market can be augmented, we focus on GFT as it is the more meaningful measure. As an extreme example, simply augmenting the market with additional sellers, and leaving them untraded, would increase the welfare while the GFT remains unchanged.
% It is not hard to see that our result implies that the welfare achieved by STR in the augmented market exceeds the first-best welfare in the original market (see \Cref{cor:main_str}).

\subsection{Our results}
\label{subsec:results}
We summarize prior results and our results in Table~\ref{table:results}.
Our first main result is \Cref{thm:main_str} which states that if we use a simple and prior-independent mechanism, namely Seller Trade Reduction (STR)\footnote{The STR mechanism was introduced by \cite{BabaioffGG20} and is a variant of the Trade Reduction mechanism \cite{McAfee92}.}, then augmenting both sides of the market by a constant number of participants has at least as much GFT as the optimal allocation before augmentation,
assuming that the buyers' distribution first-order stochastically dominates the sellers' distribution.
A formal definition of the mechanisms can be found in \Cref{sec:prelim}.

%The STR mechanism was introduced by \cite{BabaioffGG20} and is a variant of the Trade Reduction mechanism \cite{McAfee92}. Given the buyer's bids $b^{(1)} \geq \ldots \geq b^{(m)}$ and the seller's bids $s^{(1)} \leq \ldots \leq s^{(n)}$, let $r = \max\{ i \leq \min\{m ,n\} \,:\, b^{(i)} \geq s^{(i)} \}$. The welfare-maximizing allocation (called first best, which we denote by OPT) trades the $r$ highest-value buyers with the $r$ lowest-value sellers. We refer to the value $r$ as the \emph{optimal trade size}. The STR mechanism works as follows.
%For notation, let $s^{(n+1)} = \infty$.
%If $b^{(r)} \geq s^{(r+1)}$, the mechanism offers a price of $s^{(r+1)}$ to the $r$ highest-value buyers with the $r$ lowest-value sellers, resulting in a matching of size $r$. Otherwise, the $r-1$ highest-value buyers trade with the $r-1$ lowest-value sellers. Each traded buyer pays $b^{(r)}$ while each traded seller receives $s^{(r)}$. In the latter case, we say that the last trade is \emph{reduced}.

\newcommand{\withincell}[1]{
\begin{tabular}{@{}c@{}}
    #1
\end{tabular}
}

\begin{table}[t]
    % \centering
    \centerline{
    \begin{tabular}{|c|c|c|c|}
    \hline 
    \multirow{2}{*}{Distribution Assumptions} & \multicolumn{2}{c|}{Upper Bound} & \multirow{2}{*}{Lower Bound} \\ 
    \cline{2-3}
    & Previous Work & This Work & \\
    \hline
    $F_B=F_S$ & \withincell{1 buyer \\ \cite[Theorem~1.1]{BabaioffGG20}} & -- & 1 agent (trivial) \\ 
    \hline
    $F_B$ FSD $F_S$ & \withincell{$n(m+4\sqrt{m})$ buyers \\ \cite[Theorem~1.10]{BabaioffGG20}} & \withincell{$O(1)$ agents \\ (\Cref{thm:main_str})} & \withincell{$\lfloor \log_2 m \rfloor$ buyers \\ \cite[Theorem~5.1]{BabaioffGG20}} \\  
    %\hline
    %$\Pr_{b\sim F_B, s\sim F_S}[b\geq s] = r$ & -- &
    %\begin{tabular}{@{}c@{}} $O(\frac{\log(m/n)+\log(1/r)}{r})$ agents \\ (Theorem~\ref{thm:non_fsd}) \end{tabular}
    % & \begin{tabular}{@{}c@{}}
    % $\Omega(r^{-1/8})$ agents\\ (\Cref{thm:non_fsd_lb_intro})
    % \end{tabular} \\
    \hline
    No assumption. & -- & -- & \withincell{any finite number \\ \cite[Proposition~3.4]{BabaioffGG20}}\\
    \hline
    \withincell{$F_{B,i}^{-1}(1-\gamma) \geq F_{S,j}^{-1}(\gamma)$ \\ for new buyer $i$, new seller $j$. \\ No assumptions on original market.} & -- & \withincell{$O(1/\eps \gamma^2)$ agents for \\ $(1-\eps)$-approximation \\ (\Cref{thm:fsd_noniid_approx})} & \withincell{$\Omega(1/\eps \gamma)$ agents for \\ Trade Reduction \\ (\Cref{prop:approx_lb})} \\
    \hline
    \end{tabular}
    }
    \caption{Summary of our main results. The upper bounds state the number of additional agents suffices for a prior-independent mechanism (BTR or STR) to achieve GFT at least the first best. The lower bounds state the number of additional agents necessary for any anonymous and deterministic mechanism. Results with ``buyer'' listed indicate that only buyers can be augmented to the market. Note that the last row is for an approximation result instead of beating the GFT of the first best.}\label{table:results}
\end{table}

\begin{restatable}{theorem}{mainstr}
\label{thm:main_str}
Consider the double auction with $m$ i.i.d.~buyers and $n$ i.i.d.~sellers. Suppose $m\geq n$ and the buyers' distribution $F_B$ first-order stochastically dominates the sellers' distribution $F_S$. Then there is a global constant integer $c>0$ 
such that 
the GFT of STR with $m+c$ buyers and $n+c$ sellers 
is at least the first-best GFT with $m$ buyers and $n$ sellers.
\end{restatable}

\begin{remark}\label{remark:mn_assumption}
While Theorem~\ref{thm:main_str} assumes that $m \geq n$, the result applies analogously to the case where $m\leq n$ using Buyer Trade Reduction, by negating the values/costs and swapping the role of buyers and sellers. See \cite[Proposition~A.1]{BabaioffGG20}.
%one can also handle the case where $m \leq n$ using Buyer Trade Reduction instead (see \cite[Proposition~A.1]{BabaioffGG20}).
\end{remark}
\begin{remark}
    Another natural benchmark is to consider the per-unit GFT defined as the GFT divided by the number of items in the market.
    In the setting where we recruit only buyers, as in \cite{BabaioffGG20}, the per-unit GFT objective is equivalent to the GFT objective.
    However, the per-unit GFT objective is a strictly stronger benchmark when one is also allowed to recruit sellers.
    Thus, it is natural to ask if recruiting $O(1)$ agents suffices for this stronger benchmark.
    In Appendix~\ref{app:per_unit_gft}, we build on the lower bound example in \cite{BabaioffGG20} to prove that if there are $m$ buyers in the original market, $\Omega(\log m)$ agents are necessary for the per-unit GFT in the augmented market to exceed the per-unit GFT in the original market.
\end{remark}

Since the welfare of any mechanism is the GFT plus the sum of the seller values, our results immediately apply to the welfare objective, as the sum of seller values in the augmented market is at least the one in the original market.
% \begin{corollary}
% \label{cor:main_str}
% Consider the double auction with $m$ i.i.d. buyers and $n$ i.i.d. sellers. Suppose $m\geq n$ and the buyers' distribution $F_B$ first-order stochastically dominates the sellers' distribution $F_S$. Then there is a global constant integer $c>0$ 
% such that 
% the \emph{welfare} of STR with $m+c$ buyers and $n+c$ sellers 
% is at least the welfare of the optimal allocation with $m$ buyers and $n$ sellers.
% \end{corollary}

%We also extend \Cref{thm:main_str} to distributions that may not satisfy the FSD condition. Existing hardness results show that no finite bound applies to any pair of distributions even for the single seller case~\cite{BabaioffGG20}. Instead we prove a parameterized result using the parameter $r = \probs{b \sim F_B, s \sim F_S}{b \geq s} \in (0,1)$. See the third row of Table~\ref{table:results}.

For our second main result, we consider a setting where we make no assumptions on the original market and only fairly mild assumptions on the recruited agents' distributions.
\begin{restatable}{theorem}{fsdapprox}
    \label{thm:fsd_noniid_approx}
    Fix $\gamma \in (0, 1/2]$ along with a set of $c$ buyers with value distributions $F_{B,1}, \ldots, F_{B,c}$ and a set of $c$ sellers with value distributions $F_{S,1}, \ldots, F_{S,c}$ such that all agents' values are mutually independent and for all $i,j \in [c]$, we have $F_{B,i}^{-1}(1-\gamma) \geq F_{S,j}^{-1}(\gamma)$.
    Fix any market $M$ with arbitrary correlation between buyers and sellers.
    Suppose that we augment $M$ by including the $c$ buyers and $c$ sellers described above.
    % recruiting $c$ additional buyers with value distributions $F_{B,1}, \ldots, F_{B,c}$ and $c$ additional sellers with value distributions $F_{S,1}, \ldots, F_{S,c}$, that are mutually independent and independent of $M$.
    % Further, assume that for every $i,j \in [c]$, we have $F_{B,i}^{-1}(1-\gamma) \geq F_{S,j}^{-1}(\gamma)$ for some $\gamma \in (0, 1/2]$.
    Let $M'$ be the augmented market.
    Then the GFT of Trade Reduction is at least a $(1-O(1/\gamma^2c))$-approximation to the GFT of the optimal matching in $M'$.
\end{restatable}
We stress that \Cref{thm:fsd_noniid_approx} makes no assumptions on $M$ and that the value distributions of the agents that we augment into the market is completely agnostic of $M$.

There are several ways that one can interpret this result.
The most obvious is that simply recruiting agents into the market suddenly makes a simple mechanism efficient.
For example, Uber can simply recruit more drivers and more riders into the platform without any further market analysis.
Next, in any large market, it is reasonable to assume that there must be a small subset of buyers with high value and a small subset of sellers that can produce goods at relatively low values.
Our result implies that a simple mechanism is already efficient.

\begin{remark}
    Earlier, we stated that a sufficient condition for Trade Reduction in an augmented market to obtain a $(1-\eps)$-approximation is that the buyers' value exceeds the sellers' value with constant probability.
    We note that this condition implies the c.d.f.~condition in \Cref{thm:fsd_noniid_approx}, up to a constant.
    Indeed, if buyer $i$'s value exceeds seller $j$'s value with probability at least $\gamma$ then $F_{B,i}^{-1}(1-\gamma/2) \geq F_{s,j}^{-1}(\gamma /2)$ (see \cite[Lemma~3.1]{CaiGMZ21}). %or \Cref{lemma:r_fsd}).
\end{remark}

\comment{
% does not imply that the optimal trade size increases.
% The issue in the previous example is that even though the new buyer and new seller could trade, including them with the original buyers and sellers causes the new buyer to ``steal'' an original trade.
% This means that the size of the optimal trade in the augmented market remains the same and STR could potentially have a trade size that is one smaller than this.
% which causes the size of optimal trade in the augmented market to remain the same and thus STR has one less trade.
%while causing us to lose a trade since we are using STR.
To get around this issue, it is helpful to find a more structured event that ensures 
% that when we do have a trade among a new buyer and a new seller,
that including new agents increases the trade size of the optimal matching, and so that it can offset the loss of GFT in the event where STR has a smaller trade size than OPT.
% Note that STR loses at most one trade and therefore the GFT of STR is at least that of OPT in this case.
% To ensure that STR is strictly better than OPT, we need to offset the loss of GFT in the event where STR has a smaller trade size than OPT.
\todo{We can mention here that there is a tension between defining high probability good events and being able to lower bound the condition mean of the gain of this event. E.g.~take $\cE_1$ to be the event where we increase the trade size by $2$. Hard to estimate.}
\todo{Actually, tension may not be super clear to the reader. We can say that if we define naive events then we cannot get a good understanding of the conditional mean. So we need to define more structured events.
}
For this reason, we consider an event that increases the trade size by at least two.
To make this formal, we construct a ``good'' event $\cE_1$ where STR outperforms
% While the previous paragraph suggests how we might come up with a way to obtain strictly more trades with STR than in OPT,
% one still needs to quantitatively bound the GFT when STR performs better than OPT and show that this offsets the loss of GFT when STR performs worse than OPT.
% A key insight in our proof is to craft a ``good'' event $\cE_1$ where STR outperforms
OPT by at least $C$ (in expectation) and a ``bad'' event $\cE_2$ where STR underperforms OPT by at most $C$ (in expectation).
For the good event $\cE_1$, one naive definition is to define it as the event that the optimal trade size increases by two. However, it is difficult to understand the structure of this event and provide a good quantitative lower bound on the conditional expectation of the increase in GFT.
Instead, we define a more structured event $\cE_1$ 

For the good event $\cE_1$, we ensure that (i) at least two new buyers have value which is greater than the value of some originally traded buyer and (ii) at least two new sellers have value which is less than the value of some originally traded seller.
As we show, this is sufficient for the optimal trade size to increase by two.
We then define the event $\cE_2$ such that it is a necessary condition for STR to have a smaller trade size than OPT.
An obvious choice is to define $\cE_2$ as $\overline{\cE_1}$.
When $n \approx m$, we can show that $\prob{\cE_1} \geq 1/2$.
Therefore, this definition of $\cE_2$ is sufficient for our purposes.
However, when $n \ll m$, $\prob{\cE_1} \ll 1/2$ and thus, $\overline{\cE_1}$ is too large of an event.
Roughly speaking, we choose $\cE_2$ so that all the new sellers have value which is larger than the lowest value originally traded buyer.

\chris{Mention the contention between probability and gain.}

When $n \approx m$, this is a fairly high probability event since 

In this case, we define $\cE_2$ as the event where the value of all the new sellers are among the top $2n$ agent values.
We show that outside of this event, STR performs at least as well OPT.
When $n \approx m$, this is a trivial event but when $n \ll m$, this event is very small.

However, we can make this event smaller by also enforcing that most 

The bad event $\cE_2$ is that all the new sellers have very high value.
We show that if even one seller does not have too large of a value then STR performs at 

Simultaneously, we will guarantee that $\prob{\cE_1} \geq \prob{\cE_2}$.

% \textcolor{red}{Comments below.}

% 1. Uber example
% 2. Welfare vs. GFT

% Mention that a technical difficulty is that it is \emph{not} sufficient that to just have additional possible trades.

% Condition 1: The intuition is that we get a new buyer and a new seller that is able to trade. Intuitively, this is sufficient. However, it turns out that this is \emph{not} sufficient.

% Condition 2: Even if we do gain one more trade, we have to be careful since it is not clear that this additional gain will recover the loss when we do lose the trade.

% Mention some roadblocks with the analysis of using trade reduction vs seller trade reduction.

% A few questions:
% - Could you do this with just augmenting the smaller side by a constant number?
% % - For $(m, \Theta(m))$ case in Babaioff, did they really need

% Intuition is that only one additional trade is needed as long as neither side is saturated.
% This is true for $(m, m)$ case and if this intuition is true then we ought to be able to only augment one side of the market.
% But even here, best upper bound is $O(m^2)$ due to Yannai.
}
\subsection{Additional Related Work}\label{subsec:related_work}

The paper that is mostly related to our work is by Babaioff, Goldner, and Gonczarowski \cite{BabaioffGG20}. They study Bulow-Klemperer-style results in two-sided market where one side of the market is augmented. When the buyer's distribution is the same as the seller's distribution, they prove that one additional buyer is sufficient for BTR to achieve welfare at least the first-best welfare in the original market. They then study the problem with the stochastic dominance assumption, proving an upper bound of $4\sqrt{m}$ for a single seller and $n(m+4\sqrt{m})$ for $n$ sellers. They also provide lower bounds on the number of additional buyers required. Their lower bounds apply not only to BTR and STR, but also to any deterministic and prior-independent mechanisms. In this paper we study the same problem but allow both sides of the market to be augmented.

\vspace{-.2in}
\paragraph{Approximations in two-sided markets.}

Despite the impossibility result by Myerson and Satterthwaite \cite{MyersonS83}, many recent papers have successfully shown a multiplicative approximation to the first-best and second-best objective in various settings of two-sided markets. One line of work, which focuses on bilateral trade, aims to approximate the optimal welfare or GFT and to study the difference between the first-best and second-best \cite{BlumrosenD21,colini2017fixed, DengMSW22,KangPV22,CaiW23, BrustleCWZ17}.
% which focuses on bilateral trade aims to approximate the optimal welfare or GFT in bilateral trade, and studying the difference between the first-best and second-best.
Another line of work studies the approximation problem in more general two-sided markets such as double auctions and multi-dimensional two-sided markets~\cite{Colini-Baldeschi16, colini2020approximately, DuttingRT14, BabaioffCGZ18, CaiGMZ21}. In sharp contrast to our paper, the mechanisms in all these works {are not prior-independent: either the mechanism designer or the agents need to know the others' prior distributions.}
%depend on the prior distributions. 
%Moreover, the parameter $r$ used in our parameterized result (\Cref{thm:non_fsd}) %indicates how much the buyer distribution overlaps with the seller distribution. 
%{is the probability that the buyer's value exceeds the seller's cost.} This parameter is also useful in proving the parameterized approximation ratio in two-sided markets~\cite{CaiGMZ21, colini2017fixed}. 
Another line of work provides \emph{asymptotic} approximation guarantees in the number of items optimally traded for settings as general as multi-unit buyers and sellers and $k$ types of items \cite{McAfee92,SegalHaleviHA18a,SegalHaleviHA18b, BabaioffCGZ18}. Moreover, \cite{MaoPW22} consider a model of interactive communication in bilateral trade and prove that the efficient allocation is achievable with a smaller number of rounds of communication.

\vspace{-.2in}
\paragraph{Bulow-Klemperer-style results in one-sided markets.}

There have been many Bulow-Klemperer-style results that aim to beat or approximate the optimal revenue in auctions with the recruitment of additional buyers. Results in single-dimensional settings include \cite{HartlineR09, DughmiRS09, FuLR19} for regular distributions, \cite{SivanS13} for irregular distributions, and \cite{LiuP18} for a dynamic single-item auction.
%, adding 1 buyer per distribution guarantees a 3-approximation to the optimal revenue. \citet{SivanS13} extend the results to the case of irregular distributions.
Another line of work extend the results to multi-dimensional auctions, when buyers are unit-demand~\cite{roughgarden2012supply} and additive~\cite{BeyhagiW19, EdenFFTW17, FeldmanFR18, CaiS21}. Results in this paper (and \cite{BabaioffGG20}) show that Bulow-Klemperer-style results can also derived in two-sided markets. 
We note that in the revenue-maximizing auction setting, it is clearly impossible to perform augmentation while being completely agnostic to the agents' distributions. On the other hand, one of our main result is that it is possible to perform augmentation in the efficiency-maximizing two-sided market setting while being completely agnostic to the market.
\section{Preliminaries}
\label{sec:prelim}
%\paragraph{Basic notation.}
%We always use $m$ to refer to the number of buyers, $n$ to refer to the number of sellers, and $c$ to denote the number of additional buyers (and sellers) that we augment into the market. We without loss of generality assume that $m\geq n$ (see \Cref{remark:mn_assumption}). We let $N = m + n + 2c$ be the total number of agents in the augmented market.

% \textcolor{red}{Do not mention the independence assumption here; mention it when we discuss the BK result.}

% \textcolor{red}{Keep the first sentence and move the rest to BK result.}

\vspace{-.1in}
\paragraph{Double Auction and Gains From Trade.}
This paper focuses on the double auction setting, a two-sided market with $m$ unit-demand buyers and $n$ unit-supply sellers. Without loss of generality, we assume that $m\geq n$ (see \Cref{remark:mn_assumption}). 
All items are interchangeable and thus the value for each agent can be described as a scalar. 

An allocation in a double auction is a (possibly random) set of $n$ agents who hold the items. A buyer \emph{trades} in the allocation if she holds the item and a seller \emph{trades} if she does not hold the item. The \emph{gains from trade} (GFT) of an allocation is defined as the difference between the sum of all traded buyers' values and the sum of all traded sellers' values.

\vspace{-.2in}
\paragraph{Mechanisms.}
%We assume that there are $m$ buyers and $n$ sellers where $m \geq n$.
%We denote the buyer values by $b_1 \geq \ldots \geq b_m$ and the seller values by $s_1 \leq \ldots \leq s_n$.
We denote the buyer values by $b_1, \ldots, b_m$ and the seller values by $s_1, \ldots, s_n$. We let $\Bb = (b_1, \ldots, b_m)$ and $\Bs = (s_1, \ldots, s_n)$. A mechanism can be specified by, for each agents' profile $(\Bb, \Bs)$ an allocation and a payment for each agent.
%A mechanism for two-sided markets can be specified by a tuple $(\BxB, \BxS, \BpB, \BpS)$.
%Here, $\BxB \colon \bR^{m+n} \to [0,1]^{m}$, $\BxS \colon \bR^{m+n} \to [0,1]^n$ denote the allocation to the buyers and sellers, respectively and must satisfy $\|\BxB\|_1 + \|\BxS\|_1 = n$.
%In addition, $\BpB \colon \bR^{m+n} \to \bR^{m}$ denotes the payment \emph{from} the buyer and $\BpS \colon \bR^{m+n} \to \bR^{n}$ denotes the payment \emph{to} the seller.
%We let $\BxB_0 = \bfzero$ and $\BxS_0 = \bfone$, i.e.~initially none of the buyers are allocated and all the sellers are allocated.
%For a given allocation $\BxB, \BxS$, the \emph{gain from trade} (GFT) is defined as $\inner{\Bb}{\BxB - \BxB_0} + \inner{\Bs}{\BxS - \BxS_0}$.
We assume that all agents have quasi-linear utilities.
Specifically, if a buyer trades in the mechanism, her utility is her value minus the payment for her. Similarly if a seller trades, her utility is the payment she receives minus her value.
%if buyer $i$ has value $b_i$ and bids $b_i'$, while the other buyers' bids are $\Bb_{-i}$ and the sellers' bids are $\Bs$ then her utility is given by $u^{\mathsc{B}}(b_i'; (\Bb, \Bs)) = b_i \cdot x_i^{\mathsc{B}}((b_i', \Bb_{-i}), \Bs) - p_i^{\mathsc{B}}((b_i', \Bb_{-i}), \Bs)$.
%Similarly, if seller $j$ has value $s_j$ and bids $s_j'$, while the bueyrs' bids are $\Bb$ and all other sellers' bids are $\Bs_{-j}$ then her utility is given by $u^{\mathsc{S}}(s_j'; (\Bb, \Bs)) = s_j \cdot x_j^{\mathsc{S}}(\Bb, (s_j', \Bs_{-j})) + p_j^{\mathsc{S}}(\Bb, (s_j', \Bs_{-j}))$.
A mechanism is \emph{Bayesian Incentive Compatible} (BIC) if every agent maximizes her \emph{expected} utility (over all the other agents' randomness and the randomness of the mechanism) when she bids truthfully her value. In addition, it is \emph{Dominant Strategy Incentive Compatible} (DSIC) if every agent maximizes her utility when she bids truthfully, no matter what the other agents report.
%, i.e.~for every buyer $i$, $b_i \in \argmax_{b_i'} u^{\mathsc{B}}(b_i'; (\Bb, \Bs))$ for every $\Bb, \Bs$ (and analogously for the sellers).
We say that a mechanism is \emph{individually rational} (IR) if every agent has non-negative utility when she bids truthfully, no matter what the other agents report. %i.e.~$u^{\mathsc{B}}(b_i; (\Bb, \Bs)) \geq 0$ for every $\Bb, \Bs$ (and analogously for the sellers).
A mechanism is said to be \emph{weakly budget-balanced} (WBB) if the sum of payment from the buyers is at least the sum of payment to the sellers for any agents' profile, i.e. the mechanism does not run a deficit.
%the mechanism does not run a deficit 
%if for every $\Bb, \Bs$, $\sum_{i=1}^m p_i^{\mathsc{B}}(\Bb, \Bs) - \sum_{j=1}^n p_j^{\mathsc{S}}(\Bb, \Bs) \geq 0$, i.e.~the mechanism never runs a deficit.

\vspace{-.2in}
\paragraph{First Best and Trade Reduction.}
Given any buyers' profile, the \emph{first-best allocation} (also denoted by OPT) is the welfare-maximizing allocation under this profile (the allocation for the VCG mechanism). Formally, let $b^{(1)} \geq \ldots \geq b^{(m)}$ be the buyer's bids ordered in the non-increasing order and $s^{(1)} \leq \ldots \leq s^{(n)}$ be the seller's bids ordered in non-decreasing order. We abuse the notation and use $b^{(i)}$ and $s^{(i)}$ to represent the corresponding buyer and seller. The first-best allocation trades buyers $b^{(1)}, \ldots, b^{(r)}$ with $s^{(1)}, \ldots, s^{(r)}$, where $r = \max\{i \leq \min\{m, n\} \,:\, b^{(i)} \geq s^{(i)} \}$. We refer to $r$ as the \emph{optimal trade size}.
Next, we define the trade reduction that we consider in this paper.

\begin{definition}[Trade Reduction Mechanism \cite{McAfee92}]
    \label{defn:trade_reduction}
    Let $u \in [0, 1]$ be a parameter.
    If $r < \min\{m ,n\}$ and $b^{(r)} \geq u \cdot b^{(r+1)} + (1-u) \cdot s^{(r+1)} \geq s^{(r)}$ then TR trades buyers $b^{(1)}, \ldots, b^{(r)}$ with $s^{(1)}, \ldots, s^{(r)}$ at price $u \cdot b^{(r+1)} + (1-u) \cdot s^{(r+1)}$.
    Otherwise, the mechanism trades buyers $b^{(1)}, \ldots, b^{(r-1)}$ with $s^{(1)}, \ldots, s^{(r-1)}$ (if $r\leq 1$ then there is no trade). Each traded buyer pays $b^{(r)}$ and each traded seller receives $s^{(r)}$.
\end{definition}

% \modify{A general version of the Trade Reduction (TR) mechanism \cite{McAfee92} can be defined as follows. Let $u \in [0,1]$ be a parameter. If $r < \min\{m ,n\}$ and $b^{(r)} \geq u \cdot b^{(r+1)} + (1-u) \cdot s^{(r+1)} \geq s^{(r)}$ then TR trades buyers $b^{(1)}, \ldots, b^{(r)}$ with $s^{(1)}, \ldots, s^{(r)}$ at price $u \cdot b^{(r+1)} + (1-u) \cdot s^{(r+1)}$.
% Otherwise, the mechanism trades buyers $b^{(1)}, \ldots, b^{(r-1)}$ with $s^{(1)}, \ldots, s^{(r-1)}$ (if $r\leq 1$ then there is no trade). Each traded buyer pays $b^{(r)}$ and each traded seller receives $s^{(r)}$.}

Our first main result (\Cref{thm:main_str}) holds for a particular version of TR where $u = 0$ which we refer to as seller's trade reduction (STR). We note that \cite{BabaioffGG20} also consider an asymmetric version of TR where they set $u = 1$; they refer to this version as buyer's trade reduction (BTR). Our second main result (\Cref{thm:fsd_noniid_approx}) holds for all variants of TR in addition to the variant where we only utilize the ``otherwise'' part of the above mechanism. Specifically, we never trade buyer $b^{(r)}$ and seller $s^{(r)}$. Buyers $b^{(1)}, \ldots, b^{(r-1)}$ are offered a price of $b^{(r)}$ and sellers $s^{(1)}, \ldots, s^{(r-1)}$ are offered a price of $s^{(r)}$.
% In the paper we focus on a variant of the McAfee's Trade Reduction mechanism \cite{McAfee92} known as \emph{Seller Trade Reduction} (STR), which was introduced by \citet{BabaioffGG20}.
% %The mechanism is stated in Subsection~\ref{subsec:results} but we restate here.
% %Let $b_1 \geq \ldots \geq b_m$ be the buyer's bids and $s_1 \leq \ldots \leq s_n$ be the seller's bids.
% %For notation, we let $s_{n+1} = \infty$.
% %Let $r = \max\{i \leq \min\{m, n\} \,:\, b_i \geq s_i \}$ be the optimal trade size.
% The STR mechanism is stated as follows. For ease of notation, we let $s_{n+1} = \infty$. If $b^{(r)} \geq s^{(r+1)}$, then STR trades buyers $b^{(1)}, \ldots, b^{(r)}$ with $s^{(1)}, \ldots, s^{(r)}$ at price $s^{(r+1)}$. %Notice that this is exactly the first-best allocation. 
% Otherwise, the mechanism trades buyers $b^{(1)}, \ldots, b^{(r-1)}$ with $s^{(1)}, \ldots, s^{(r-1)}$ (if $r\leq 1$ then there is no trade). Each traded buyer pays $b^{(r)}$ and each traded seller receives $s^{(r)}$.
% %and 
% %the we offer a price of $s_{r+1}$ to buyers $1, \ldots, r$ and sellers $1, \ldots, r$.
% %Otherwise, we offer a price of $b_r$ to buyers $1, \ldots, r-1$ and a price of $s_r$ to sellers $1, \ldots, r-1$.
The following lemma shows that Trade Reduction is an incentive-compatible mechanism.\footnote{\cite{BabaioffGG20} prove this for STR but it is not difficult to adapt their proof of TR.}
\begin{lemma}[{\cite[Proposition~C.1]{BabaioffGG20}}]
TR is a deterministic, prior-independent mechanism and satisfies DSIC, IR, and WBB.
\end{lemma}

\section{Constant Agents Suffice to Beat First-Best when \texorpdfstring{$F_B$}{FB} FSD \texorpdfstring{$F_S$}{FS}}
\label{sec:str}
For the rest of the paper, we focus on the i.i.d. setting and study the problem of beating the first-best GFT through augmentation.
We prove that STR with $O(1)$ additional agents extracts at least as much GFT as the first-best allocation with the original set of agents (Theorem~\ref{thm:main_str}).
Throughout this section we assume that buyer (resp.~seller) values are drawn i.i.d.~according to a common cumulative density function $F_B$ (resp.~$F_S$).
% In this paper, we abuse terminology and often refer to $F_B$ and $F_S$ as distributions as well.
For any quantile $q\in (0,1)$, define the value $b(q)$ corresponding to quantile $q$ as $b(q) = \inf\{x\mid \Pr_{b\sim F_B}[b\leq x]\geq q\}$. Similarly, define $s(q) = \inf\{x\mid \Pr_{s\sim F_S}[s\leq x]\geq q\}$. Clearly both $b(q)$ and $s(q)$ are non-decreasing in $q$. We say that $F_B$ \emph{first-order stochastically dominates} (FSD) $F_S$ if for every $q \in (0, 1)$, $b(q) \geq s(q)$.

% \mainstr*

\subsection{Proof Techniques}\label{subsec:fsd_proof_techniques}
First, we present a high-level discussion about the proof techniques in this section.
Notice that STR loses no more than a single trade from the first-best allocation in the augmented market.
%Since STR loses no more than a single trade, 
Thus a natural (but erroneous) starting point to prove Theorem~\ref{thm:main_str} may be to (i) show that with only a constant number of new buyers and new sellers, at least one of the new buyers is eligible to trade with a new seller and (ii) show that if there is a trade between a new buyer and a new seller then the trade size must increase by $1$ and thus STR performs at least as well as OPT.
If the second statement were true then the proof should be relatively straightforward since the first statement happens with fairly high probability due to the stochastic dominance assumption. Unfortunately, the second statement is false and thus the first statement is not a sufficient condition for STR to outperform OPT.
For an example where this happens, see Appendix~\ref{app:str_loses_trade_example}.

The message in the previous paragraph is that having additional trades among the new agents is not sufficient to guarantee that the optimal trade size increases. We would like to find an event such that the optimal trade size increases, which is sufficient for STR to outperforms OPT.
Naively, we could simply consider the event where the optimal trade size does increase.
However, the difficulty is in being able to lower bound the gain of the expected GFT restricted to this event and compare that with the loss of the expected GFT when this does not happen.
In order to make the analysis more feasible, we consider more structured events that (i) make it possible to analyze the gain or loss in GFT and (ii) we can compare the probabilities of these events.

To make this formal, we use a coupling argument that was also used by Babaioff, Goldner, and Gonczarowski \cite{BabaioffGG20}.
We first fix a set of quantiles and then assign these quantiles uniformly at random to the new and original buyers and sellers.
However, the techniques in our paper and Babaioff, Goldner, and Gonczarowski \cite{BabaioffGG20} are otherwise very different.
Babaioff, Goldner, and Gonczarowski \cite{BabaioffGG20} first consider the single seller and $m$ buyers setting and proceed by showing that by adding a sufficient number of buyers it must be that (i) the GFT difference between the new and original optimal allocations is large and (ii) the GFT difference between the new optimal allocation and BTR is small.
The only way for this to be possible is that the GFT of BTR must be large compared to the original optimal allocation.
To handle the case with an arbitrary number of sellers, they show that they can reduce the problem to the single seller case but this reduction incurs a linear overhead (in the number of sellers).
In contrast, our argument directly compares the GFT difference between STR and OPT and show that this difference is net positive.

We now proceed with additional details on our argument.
In the augmented market, $m+c$ buyers (including $m$ original buyers and $c$ augmented buyers) draw their values i.i.d.~from $F_B$ and $n+c$ sellers (including $n$ original sellers and $c$ augmented sellers) draw their values i.i.d.~from $F_S$. Denote $N=m+n+2c$ the total number of agents in the augmented market. We notice that the distribution of $b(q)$ (resp. $s(q)$) where $q$ is drawn uniformly at random from $(0,1)$ is exactly the distribution $F_B$ (resp. $F_S$). We thus couple the random augmented market with the following random process: We draw $N$ uniform quantiles from $(0,1)$ and then assign these quantiles to all agents in the augmented market uniformly at random. 

More specifically, denote $q_1, \ldots, q_N$ the $N$ uniform quantiles in non-increasing order so that $q_1 \geq \ldots \geq q_N$. Let $\Bq = (q_1, \ldots, q_N)$.
%\chrisnote{Note that for the entirety of the proof, we \emph{condition} on a set of quantiles.}
To avoid too many subscripts, we sometimes abuse notations and use $q(i)$ to denote $q_i$.
These quantiles are assigned to all agents in the augmented market, including all original (called ``old'') and augmented (called ``new'') buyers and sellers. We notice that any two old buyers (or old sellers, new buyers, new sellers) are interchangeable, i.e. swapping their values will not change the GFT of the first-best allocation and STR in both the original and augmented market. Thus it suffices to consider any assignment from quantiles to those four labels.
Formally, let $\pi \colon [N] \to \{\BO, \BN, \SO, \SN\}$ be a function that maps (quantile) indices to old buyers, new buyers, old sellers, and new sellers, respectively. Let $\Pi_{n,m,c} = \{ \pi \,:\, |\pi^{-1}(\BO)| = m, |\pi^{-1}(\SO)| = n, |\pi^{-1}(\BN)| = |\pi^{-1}(\SN)| = c\}$ be the set of valid assignments. The assignment we choose is thus uniformly drawn from $\Pi_{n,m,c}$. 

For any fixed quantiles $\Bq$ and valid assignment $\pi$, denote $\STR(\Bq,\pi)$ the GFT of Seller Trade Reduction in the \textbf{augmented} market and denote $\OPT(\Bq,\pi)$ the GFT of the first-best allocation in the \textbf{original} market. Both values are well-defined since they are fully determined by the quantiles $\Bq$ and the assignment $\pi$.
Thus $\STR = \STR(m+c, n+c)=\bE_{\Bq,\pi}[\STR(\Bq,\pi)]$ and $\OPT = \OPT(m,n)= \bE_{\Bq,\pi}[\OPT(\Bq,\pi)]$.

% As discussed in the introduction, 
To prove that $\STR$ is at least $\OPT$, we would like to find an event such that the gain of the expected GFT (from first best to STR) restricted to this event can be lower bounded and compared with the loss of the expected GFT when the first-best allocation has more GFT than STR.
To formalize the idea, we would like to construct two events $\cE_1$ and $\cE_2$ over the randomness of the assignment $\pi$ such that:
\vspace{-.1in}
\begin{enumerate}
    \item For any $\Bq$, $\cE_1$ is a sufficient condition for $\STR(\Bq, \pi)\geq\OPT(\Bq, \pi)$. Moreover, $\bE_{\pi}[\STR(\Bq, \pi) - \OPT(\Bq, \pi)|\cE_1]\geq C(\Bq)$ for some $C(\Bq)>0$ (Lemma~\ref{lemma:E1}).
    \item For any $\Bq$, $\cE_2$ is a necessary condition for $\OPT(\Bq, \pi)>\STR(\Bq, \pi)$. Moreover, $\bE_{\pi}[\OPT(\Bq, \pi) - \STR(\Bq, \pi)|\cE_2]\leq C(\Bq)$ (Lemma~\ref{lemma:notE2_STR_gt_OPT} and Lemma~\ref{lemma: bound_loss_E2}).
    \item $\Pr_{\pi}[\cE_1]\geq\Pr_{\pi}[\cE_2]$ (Lemma~\ref{lemma:probE1_gt_probE2}).
\end{enumerate}
%\chris{Condition 2 is $\geq -C(\Bq)$??}
We notice that these conditions immediately proves \Cref{thm:main_str} since
\begin{align*}
\STR& -\OPT = \bE_{\Bq, \pi}[\STR(\Bq,\pi)-\OPT(\Bq,\pi)] \\
&\geq \bE_{\Bq}[\bE_{\pi}[\STR(\Bq, \pi) - \OPT(\Bq, \pi)|\cE_1]\cdot \Pr[\cE_1]+\bE_{\pi}[\STR(\Bq, \pi) - \OPT(\Bq, \pi)|\cE_2]\cdot \Pr[\cE_2]]\geq 0.
\end{align*}
Here the first inequality follows from Property 2. For any $\Bq$, $\STR(\Bq,\pi)\geq \OPT(\Bq,\pi)$ when $\cE_2$ does not happen and thus $\STR(\Bq,\pi) \geq \OPT(\Bq,\pi)$ on the event $\neg\cE_1 \cap \neg\cE_2$.

%\chris{I tried to add more intuition about how we define events, the bucketing, and the choice of $p$.}

{To construct the above events, we first break the set of quantiles into ``buckets''. For some $p$, let $I_1$ correspond to the indices of the top $p$ quantiles (i.e.~high value agents) and $J_1$ correspond to the indices of the bottom $p$ quantiles (i.e.~low value agents).}

{As we will see below, the event $\cE_1$ that we define ensures that the matching obtained by STR contains (i) at least one new buyer from $I_1$ and one new seller from $J_1$ and (ii) the other agents in the matching have GFT at least that of OPT.
For the time-being, suppose that there were only one new buyer from $I_1$ and one new seller from $J_1$.
Then the new buyer would be a uniform random buyer from $I_1$ and the new seller would be a uniform random seller from $J_1$.
In particular, their contribution the GFT would be roughly $\expects{i,j}{b(q_i) - s(q_j)}$; this is state formally in Lemma~\ref{lemma:E1}.
If there are multiple buyers and sellers in $I_1$ and $J_1$, respectively, then one would expect that their contribution to the GFT would only increase.
This suggests taking $C(\Bq) = \expects{i,j}{b(q_i) - s(q_j)}$.
However, we note that $p$ must be $\Theta(n)$ in order for the above argument to work.
If $p \gg n$ then it becomes unlikely that new buyers in $I_1$ would be included in the first-best matching, let alone STR.
On the other hand, if $p \ll n$ then it becomes too unlikely for new agents to actually be in $I_1$ or $J_1$.}

{Analogously, it turns out that we can always upper bound the expected loss of GFT by the above choice of $C(\Bq)$ provided $p \leq n$. For the event $\cE_2$, an obvious choice is to set $\cE_2 = \neg\cE_1$. However, when $n \ll m$, the event $\cE_2$ becomes a very high probability event. For example, if $n = O(1)$ the probability that any new agent lands in $I_1 \cup J_1$ is $O(1/m)$ and so $\prob{\cE_2}$ would be $1 - O(1/m)$. To make this event smaller, we show that another necessary condition for OPT to perform better than STR is to have all the new sellers to be assigned the top $O(n)$ quantiles. If $n \ll m$ then this is a very unlikely event and we show that it is much smaller than $\prob{\cE_1}$.}

\begin{remark}
Note that some of the proofs below require that $m$, $n$, and $m-n$ are larger than a constant. This is without loss of generality, since we can add a constant number of buyers and sellers and use the first-best GFT of the augmented market as the new benchmark.
\end{remark}

\subsection{Construction of the Events}\label{subsec:fsd_event_construction}

In this section, we construct events $\cE_1$ and $\cE_2$ that satisfy the desired properties. For any valid assignment $\pi$, we denote $\Bold^{\pi} = \pi^{-1}(\BO)$ the set of indices $i$ such that the quantile $q_i$ is assigned to an old buyer. Similarly, define $\Bnew^{\pi}$, $\Sold^{\pi}$, $\Snew^{\pi}$ as the sets for new buyers, old sellers and new sellers respectively. We omit the superscript $\pi$ when the assignment is fixed and clear from context. 

By adding a constant number of buyers and sellers, we assume that $m \geq n \geq 20$. Let $p=\left\lceil \frac{n}{10} \right\rceil \geq 2$. Define the sets
\begin{align*}
I_1 = \left\{1, \ldots, p\right\}, \quad
& I_2 = \left\{ p + 1, \ldots, 2p\right\}, \\
J_1 = \{N-p+1, \ldots, N\}, \quad 
& J_2 = \{N-2p+1, N-p\}.
\end{align*}
In other words, $I_1$ denotes the first $p$ indices, $I_2$ denote the $p$ indices after $I_1$, $J_1$ denote the last $p$ indices, and $J_2$ denote the $p$ indices before $J_1$. It is straightforward to check that when $n\geq 20$, $I_1, I_2, J_1, J_2$ are all disjoint.

\begin{claim}
    $I_1, I_2, J_1, J_2$ are all disjoint.
\end{claim}

\vspace{-.2in}
\paragraph{The good event $\cE_1$.}
Define the event $\cE_1$ as the set of valid assignments $\pi$ such that all of the properties below are satisfied:
\begin{itemize}[itemsep=0pt, topsep=0pt]
    \item $|I_1 \cap \Bnew^{\pi}| \geq 2$, i.e.~there are at least $2$ new buyers in $I_1$;
    \item $|I_2 \cap \Bold^{\pi}| \geq 1$, i.e.~there are at least $1$ old buyer in $I_2$;
    \item $|J_1 \cap \Snew^{\pi}| \geq 2$, i.e.~there are at least $2$ new sellers in $J_1$;
    \item $|J_2 \cap \Sold^{\pi}| \geq 1$, i.e.~there are at least $1$ old sellers in $J_2$.
\end{itemize}
Here is the intuition for this event. We first show that every original buyer in $I_1\cup I_2$ and every original seller in $J_1\cup J_2$ trades in the original first-best allocation (Claim~\ref{claim:all_I2J2_trades}). $|I_2 \cap \Bold^{\pi}| \geq 1$ and $|J_2 \cap \Sold^{\pi}| \geq 1$ ensure that the original first-best allocation contains at least one traded buyer from $I_2$ and one traded seller from $J_2$. On top of it, the extra conditions $|I_1 \cap \Bnew^{\pi}| \geq 2$ and $|J_1 \cap \Snew^{\pi}| \geq 2$ guarantee that the optimal trade size in the augmented market is increased by at least 2, with new buyers in $I_1$ and new sellers in $J_1$ joining in the trade. This suffices to not only show that STR has GFT more than the original first-best allocation, but also prove a lower bound on the gain using values of those new traded buyers/sellers.
Formally, we prove the following lemma, whose proof is deferred to Subsection~\ref{subsec:proof_E1}.
\begin{lemma}
\label{lemma:E1}
Fix any $\Bq$.
We have that $\STR(\Bq, \pi)\geq\OPT(\Bq, \pi)$ for all $\pi\in \cE_1$. Moreover, $\bE_{\pi}[\STR(\Bq, \pi) - \OPT(\Bq, \pi)|\cE_1]\geq \expects{i, j}{b(q_i) - s(q_j)}$ where $i \sim I_1, j \sim J_1$ uniformly at random.
\end{lemma}

\vspace{-.2in}
\paragraph{The bad event $\cE_2$.}
Next, we define the bad event $\cE_2$ as $\neg \cE_1 \cap \{\pi\in \Pi_{n,m,c}\mid\Snew^{\pi} \subseteq [2n+2c]\}$.
In other words, event $\cE_2$ requires that (i) $\cE_1$ does not happen and (ii) all new sellers are in the top $2n+2c$ quantiles.
Lemma~\ref{lemma:notE2_STR_gt_OPT} shows that $\cE_2$ is a \emph{necessary} condition for OPT to obtain (strictly) more GFT than STR. We point out that $\cE_2$ is not a necessary condition for OPT to outperform the classic Trade Reduction; an example can be found in Appendix~\ref{app:compare_tr_str}.
Thus having STR is necessary for our proof.
% \footnote{A key observation we use is that, if (i) the optimal allocation in the augmented market, $\OPT'$ is not the same as the optimal allocation in the original market $\OPT$ and (ii) the size of the optimal matching remains the same then $\STR$ and $\OPT'$ have the same GFT. This would not be true using McAfee's Trade Reduction mechanism \cite{McAfee92}. For example, consider a scenario with one original buyer with value $1$, one original seller with value $0.9$, one new buyer with value $0$, and one new seller with value $0.8$. In this example, the optimal matching changes in the augmented market but remains the same size. STR would set a price of $0.9$ to the original buyer and seller while McAfee's Trade Reduction would set a price of $(0 + 0.9) / 2 = 0.45 < 0.8$. In particular, STR would retain the trade while Trade Reduction would not. With this example, it is not too difficult to find an example which shows that $\cE_2$ is not a necessary condition for OPT to outperform the classic Trade Reduction.}
%\chris{Added footnote}
\begin{lemma}
    \label{lemma:notE2_STR_gt_OPT}
    Fix any $\Bq$, we have $\STR(\Bq, \pi) \geq \OPT(\Bq, \pi)$ for all $\pi \notin \cE_2$.
\end{lemma}

Next, we bound in Lemma~\ref{lemma: bound_loss_E2} the loss in GFT conditioned on $\cE_2$, to match the lower bound proved in \Cref{lemma:E1}. To prove the lemma we use the following simple observation. The GFT loss between the original first best and STR is at most the loss between the augmented first best and STR, which is the value difference between the smallest traded buyer and the largest traded seller in the augmented market.

\begin{lemma}\label{lemma: bound_loss_E2}
  %  $\expectg{\OPT - \STR}{\cE_2} \leq \expects{i, j}{b(q(i)) - s(q(j))}$ where $i \sim I_1, j \sim J_1$ uniformly at random.
For any $\Bq$, we have $\bE_{\pi}[\OPT(\Bq, \pi) - \STR(\Bq, \pi)|\cE_2]\leq \expects{i, j}{b(q_i) - s(q_j)}$ where $i \sim I_1, j \sim J_1$ uniformly at random.
\end{lemma}
The proofs of Lemma~\ref{lemma:notE2_STR_gt_OPT} and Lemma~\ref{lemma: bound_loss_E2} can be found in \Cref{subsec:notE2_STR_gt_OPT} and \Cref{subsec:proof_bound_loss_E2}, respectively.

\vspace{-.2in}
\paragraph{Comparing probabilities of $\cE_1$ and $\cE_2$.}
To complete the proof, it remains to show that $\prob{\cE_1} \geq \prob{\cE_2}$.
%Finally, it remains to show that $\prob{\cE_1} \geq \prob{\cE_2}$.
For intuition, we consider two extremes.
First, suppose that $n = m$, i.e.~there are an equal number of buyers and sellers. Recall that $|I_1| = |I_2| = |J_1| = |J_2| = p = \left\lceil \frac{n}{10} \right\rceil$.
Assuming that $m \gg c$, we would have $|I_1| / N \approx 1/20$.
In other words, if we take a random new buyer and assign it a uniformly random index from $[N]$, then with probability roughly $1/20$ it would land in $I_1$. Since there are $c$ new buyers, we have that $\expect{|I_1 \cap \Bnew^{\pi}|} \approx c / 20 \geq 2$ provided that $c \geq 40$. Thus by concentration, if $c$ is a sufficiently large constant, then we expect that $|I_1 \cap \Bnew^{\pi}| \geq 2$ %with very high probability.
with probability at least $1-\eps$ for some small constant $\eps>0$.
Similarly, we would have $|I_2 \cap \Bold^{\pi}| \geq 1$, $|J_2 \cap \Sold^{\pi}| \geq 1$, and $|J_1 \cap \Snew^{\pi}| \geq 2$ %with high probability.
each with probability at least $1-\eps$. By union bound the good event $\cE_1$ happens %with high probability, it must be that 
with probability at least $1-4\eps$ while the bad event $\cE_2 \subseteq \neg \cE_1$ happens with probability at most $4\eps$.
This proves $\Pr[\cE_1]\geq \Pr[\cE_2]$ when $n=\Theta(m)$ (\Cref{lemma:large_n_case}).

Now, let us consider the other extreme where $n \ll m$.
In this case $|I_1| / N \approx \Omega(n / m)$ (and similarly for $I_2, J_1, J_2$).
For any fixed agent, a random assignment would land the agent in $I_1$ with probability $\Omega(n / m)$.
Thus, the probability of $|I_1 \cap \Bnew^{\pi}| \geq 2$ is $\Omega((n/m)^2)$. Similarly, the probability of $|J_1 \cap \Snew^{\pi}| \geq 2$ is $\Omega(n/m)^2$. Moreover,  the probability of the events $|J_2 \cap \Sold^{\pi}| \geq 1$ and $|J_2 \cap \Bold^{\pi}| \geq 1$ are both $\Omega(n/m)$.
Note that this is a very conservative estimate obtained by considering the event that these quantities are equal to $1$.
% \mingfeinote{(this is a conservative calculation by considering the size being exactly 1)}.
We show that the probability of $\cE_1$ is at least the product of the probabilities of the four events, which indicates that $\prob{\cE_1} = \Omega(n/m)^6$ (see \Cref{claim:independence_lb}).

On the other hand, the bad event $\cE_2$ is a subset of the event that all the new sellers are in the top $2n+2c$ quantiles.
The probability that a new seller receives a uniform index and lands in $[2n+2c]$ is $(2n+2c)/(m+n+2c) = \Theta(n/m)$. Thus, the probability that all the new sellers land in $[2n+2c]$ is $\Theta((n / m)^c)$.
Thus for a sufficiently large constant $c$, we have $\prob{\cE_2} \leq \prob{\cE_1}$ (\Cref{lemma:small_n_case}).
A formal proof of \Cref{lemma:probE1_gt_probE2} can be found in \Cref{subsec:proof_probE1_gt_probE2}.
\begin{lemma}
    \label{lemma:probE1_gt_probE2}
    Fix $c \geq 20000$ and suppose that $m \geq n + 2c$ and $n \geq c$.
    Then $\prob{\cE_1} \geq \prob{\cE_2}$.
\end{lemma}

\begin{proof}[Proof of \Cref{thm:main_str}]
It follows from Lemma~\ref{lemma:E1}, Lemma~\ref{lemma:notE2_STR_gt_OPT}, Lemma~\ref{lemma: bound_loss_E2}, and Lemma~\ref{lemma:probE1_gt_probE2}.
\end{proof}
% \section{Approximation and Augmentation for FSD but Non-i.i.d.~Distributions}
\section{Market Agnostic Recruitment}
\label{sec:fsd_noniid_approx}
% In this section, we consider the following setting.
% Let us assume that we have $m$ buyers whose value distributions are given by $F_1, \ldots, F_m$ and $n$ sellers whose cost distributions are given by $G_1, \ldots, G_n$.

% \begin{definition}
%     We say that a market is $(c,\gamma)$-dominant if there exists sets $B \subseteq [m]$ and $S \subseteq [n]$ such that
%     $|B| = |S| = c$ and $F_i^{-1}(1-\gamma) \geq G_j^{-1}(\gamma)$ for all $i \in B, j \in S$ and $\{F_i\}_{i \in B}, \{G_j\}_{j \in S}$ are mutually independent and also independent of the remainder of the market.
% \end{definition}
% We remark that the remainder of the market may have arbitrary correlation within itself.

% \begin{theorem}
%     \label{thm:fsd_noniid_approx}
%     If a market is $(c, \gamma)$-dominant then Trade Reduction obtains a $(1-O(1/\gamma^2 c))$-approximation to the first-best GFT.
% \end{theorem}

% \textcolor{red}{Add some introductory stuff here?}

% \begin{theorem}
In this section, we prove that to obtain any constant approximation to the original market, it suffices to augment the market by a constant number of buyers and sellers, satisfying some mild conditions, and run the Trade Reduction mechanism.
% \fsdapprox*

A well-known observation is that the Trade Reduction mechanism loses at most one trade compared to the optimal allocation. Moreover, the trade that is lost is the least valuable trade. Thus, \emph{if} the optimal allocation had at least $k$ trades then the Trade Reduction mechanism is a $(1-1/k)$-approximation to the optimal GFT. However, this is a conditional result and does not necessarily imply that the Trade Reduction mechanism is a good approximation to the optimal GFT.

In order to turn this conditional observation into a true approximation result, it would be sufficient to prove that the optimal GFT comes mainly from instances where there are a lot of trades. We do this using a coupling argument. Namely, for every instance $I$ that \emph{may} have a small number of trades, we map this instance into many instances $f(I)_1, \ldots, f(I)_T$ that \emph{certainly} have a large number of trades and where for each $t \in [T]$, the optimal GFT in $f(I)_t$ exceeds the optimal GFT in $I$. A technical step here is that it is not sufficient to simply have $T$ to be large; we require that the probability that we obtain the instance $I$ to be much smaller than the probability of obtaining \emph{at least one} of the instances $f(I)_1, \ldots, f(I)_T$. We prove in \Cref{lemma:combinatorial_fsd_approx} that such a mapping does exist. To summarize, we essentially show that (i) with high probability, we receive an instance where Trade Reduction is a good approximation to the optimal GFT and (ii) receiving an instance where Trade Reduction may not be a good approximation is a low probability event.

We now formalize the above argument.
First, we require the following combinatorial lemma whose proof can be found in \Cref{subsec:combinatorial_fsd_approx}.
\begin{lemma}
    \label{lemma:combinatorial_fsd_approx}
    There are functions $\alpha(\gamma) = \Theta(\gamma^2)$ and $C(\gamma) = \Theta(1/\gamma^2)$ such that the following holds.
    For any $\gamma \in (0, 1/2]$, if $c \geq C(\gamma)$ then there exists $T$ and a function $f \colon \binom{[c]}{\leq \alpha(\gamma) \cdot c} \to (2^{[c]})^T$ satisfying the following properties.
    \begin{enumerate}[topsep=0pt, itemsep=0pt]
        \item For every $t \in [T]$ and $S \in \binom{[c]}{\leq \alpha(\gamma) \cdot c}$ we have $|f_t(S)| \geq \gamma c / 2$.
        \item For every $t_1, t_2 \in [T]$ and $S_1, S_2 \in \binom{[c]}{\leq \alpha(\gamma) \cdot c}$, we have $f_{t_1}(S_1) \neq f_{t_2}(S_2)$ whenever $(t_1, S_1) \neq (t_2, S_2)$.
        \item For every $S \in \binom{[c]}{\leq \alpha(\gamma) \cdot c}$, we have $c \cdot \gamma^{|S|} (1-\gamma)^{c - |S|} \leq \sum_{t \in [T]} \gamma^{|f_t(S)|} (1-\gamma)^{c-|f_t(S)|}$.
    \end{enumerate}
\end{lemma}

For the proof, we need to define a bit of notation.
We fix $\gamma \in (0, 1/2]$ and let $\alpha, C, T, f$ be as given by \Cref{lemma:combinatorial_fsd_approx}.
Note that these parameters depend on $\gamma$ but since $\gamma$ is fixed for the proof, we omit the dependence on $\gamma$.
Let $m$ be the number of buyers in the original market and $n$ be the number of sellers in the original market.
We index the agents such that buyers $1, \ldots, c$ and sellers $1, \ldots, c$ are the new agents.
Let $F_{B,1}, \ldots, F_{B,c}$ be the value distributions for the new buyers and $F_{S,1}, \ldots, F_{S,c}$ be the distributions for the new sellers. We note that they are mutually independent and independent of the distribution of the original market.

For a set of quantiles $\Bq_B = (\qb(1), \ldots, \qb(m+c)), \Bq_S = (\qs(1), \ldots, \qs(n+c))$, define the random sets $B_+ = \{i \in [c] \,:\, \qb(i) \geq 1-\gamma\}$ and $S_+ = \{j \in [c] \,:\, \qs(j) \leq \gamma\}$.
We also define four events.
\begin{align*}
    & \cE(1,1) = \{|B_+| \geq \alpha c, |S_+| \geq \alpha c\},
    &&
    \cE(1,0) = \{|B_+| \geq \alpha c, |S_+| < \alpha c\}, \\
    & \cE(0,1) = \{|B_+| < \alpha c, |S_+| \geq \alpha c\}, 
    &&
    \cE(0,0) = \{|B_+| < \alpha c, |S_+| < \alpha c\}.
\end{align*}
% Also define the events
% \[
%     \cE(B_+, S_+) = \{\qb(i) \geq 1-\gamma \text{ if and only if } i \in B_+, \qs(j) \geq 1-\gamma \text{ if and only if } j \in S_+\}.
% \]
Finally, for sets $B' \subseteq [c], S' \subseteq [c]$, we write
\[
    \OPT(B', S') = \expectg{\OPT(\Bq_B, \Bq_S)}{B_+ = B', S_+ = S'}.
\]
We define $\TR(B', S')$ in a similar fashion.
We also write $\OPT(\cE(1,1)) = \expect{\OPT(B_+, S_+) \cdot \ind{\cE(1, 1)}}$ and similarly for $\TR(\cE(1, 1))$ and the other events $\cE(i, j)$.

First, we have the straightforward observation that the optimal GFT is monotone in the set of buyers whose quantiles are above $1-\gamma$ and the set of sellers whose quantiles are below $\gamma$.
The proof can be found in Appendix~\ref{app:other_proofs}.
\begin{lemma}
    \label{lemma:opt_monotone}
    If $B'' \supseteq B'$ and $S'' \supseteq S'$ then $\OPT(B'', S'') \geq \OPT(B', S')$.
\end{lemma}

The following lemma is a well-known and follows from a simple observation that Trade Reduction loses the least valuable matching.
\begin{lemma}
    \label{lemma:tr_lose_one_trade}
    Let $k = \min\{|B'|, |S'|\}$.
    Then $\TR(B', S') \geq \left( 1 - \frac{1}{k} \right) \OPT(B', S')$.
\end{lemma}

\begin{lemma}
    \label{lemma:opt_approx}
    If $c \geq C$ then $\OPT(\cE(1,1)) \geq (1-3/c) \cdot \OPT$.
\end{lemma}
% \textcolor{blue}{(Yang: Do we mean $c\geq N(\gamma)$?)}
% \textcolor{red}{(Chris: Yes. I added some text earlier to say we are suppressing the dependence on $\gamma$. (Will also change the $N$ to $C$, etc.))}
\begin{proof}
    Fix any $B' \subseteq B$ and $S' \subseteq S$ and let $\cE(B', S') = \{B_+ = B', S_+ = S'\}$.
    Note that $\prob{\cE(B', S')} = \gamma^{|B'|}(1-\gamma)^{c-|B'|} \gamma^{|S'|} (1-\gamma)^{c-|S'|}$.
    We now consider three cases.

    \textbf{Case 1: $|B'| < \alpha c$ and $|S'| < \alpha c$.~~} By \Cref{lemma:combinatorial_fsd_approx}, we have
    \begin{align*}
        \sum_{t_1, t_2 \in [T]} & \OPT(f_{t_1}(B'), f_{t_2}(S')) \cdot \prob{\cE(f_{t_1}(B'), f_{t_2}(S'))} \\
        & = \sum_{t_1, t_2 \in [T]} \OPT(f_{t_1}(B'), f_{t_2}(S')) \cdot \gamma^{|f_{t_1}(B')|}(1-\gamma)^{c-|f_{t_1}(B')|} \gamma^{|f_{t_2}(S')|} (1-\gamma)^{c-|f_{t_2}(S')|} \\
        & \geq c \cdot \sum_{t_1 \in [T]} \OPT(f_{t_1}(B'), S') \cdot \gamma^{|f_{t_1}(B')|}(1-\gamma)^{c-|f_{t_1}(B')|} \gamma^{|S'|} (1-\gamma)^{c-|S'|} \\
        & \geq c^2 \cdot \OPT(B', S') \cdot \gamma^{|B'|}(1-\gamma)^{c-|B'|} \gamma^{|S'|} (1-\gamma)^{c-|S'|} \\
        & =
        c^2 \cdot \OPT(B', S') \cdot \prob{\cE(B', S')}.
    \end{align*}
    In particular, the first inequality uses \Cref{lemma:combinatorial_fsd_approx} with $S$ replaced by $S'$ and \Cref{lemma:opt_monotone} to show that $\OPT(f_{t_1}(B'), f_{t_2}(S')) \geq \OPT(f_{t_1}(B'), S')$.
    The second inequality is similar which uses \Cref{lemma:combinatorial_fsd_approx} with $S$ replaced by $B'$ and \Cref{lemma:opt_monotone} to show that $\OPT(f_{t_1}(B'), S') \geq \OPT(B', S')$.
    Observe that the first line is a lower bound on $\OPT(\cE(1, 1))$ (this uses the second assertion of \Cref{lemma:combinatorial_fsd_approx}).
    % \textcolor{blue}{(Yang: Let's explain why the two inequalities are true.)} \textcolor{red}{(Chris: is it clearer now?)}
    Thus, we can conclude that $\OPT(\cE(1, 1)) \geq c^2 \cdot \OPT(\cE(0, 0))$.

    \textbf{Case 2: $|B'| < \alpha c$ and $|S'| \geq \alpha c$.~~} The calculation is similar to the first case. By \Cref{lemma:combinatorial_fsd_approx} and \Cref{lemma:opt_monotone}, we have
    \begin{align*}
        \sum_{t \in [T]} \OPT(f_{t}(B'), S') \cdot \prob{\cE(f_{t}(B'), S')}
        \geq
        c \cdot \OPT(B', S') \cdot \prob{\cE(B', S')}.
    \end{align*}
    We conclude that $\OPT(\cE(1, 1)) \geq c \cdot \OPT(\cE(0, 1))$.

    \textbf{Case 3: $|B'| \geq \alpha c$ and $|S'| < \alpha c$.~~}
    This is analogous to the previous case and we get that $\OPT(\cE(1, 1)) \geq c \cdot \OPT(\cE(1, 0))$.
\end{proof}

\begin{proof}[Proof of \Cref{thm:fsd_noniid_approx}]
    Note that on the event $\cE(1, 1)$, the optimal trade size is at least $\alpha c$ and thus by \Cref{lemma:tr_lose_one_trade},
    we have $\TR \geq \TR(\cE(1, 1)) \geq \left(1 - 1/\alpha c \right) \OPT(\cE(1, 1))$.
    Next, by \Cref{lemma:opt_approx}, we have $\OPT(\cE(1, 1)) \geq (1 - 3 / c) \cdot \OPT$.
    We conclude that
    \[
        \TR \geq (1-1/\alpha c) \cdot (1-1/3c) \cdot \OPT \geq (1 - (3+1/\alpha)/c) \cdot \OPT.
    \]
    Recalling that $\alpha = \Theta(\gamma^2)$ completes the proof.
\end{proof}

\section{Summary}\label{sec:summary}

In this paper we prove Bulow-Klemperer-style results in two-sided markets. When the buyer distribution FSD the seller distribution, we show that a deterministic, DSIC, IR, BB and prior-independent mechanism with constant additional agents achieved GFT at least the first-best GFT in the original market. Here a constant number of buyers and sellers are both added to the market. While Babaioff, Goldner, and Gonczarowski \cite{BabaioffGG20} study the problem where only the larger side of the market is augmented (buyers are augmented with the assumption of $m\geq n$), it is an interesting direction to study the problem where only the smaller side of the market is allowed to augment. Intuitively, augmenting to the smaller side of the market is more efficient in increasing the trade size and GFT. Results in this direction yet remain open.
Finally, we prove that adding independent agents agnostic to the (arbitrarily correlated) original market such that $F_{B,i}^{-1}(1-\gamma) \geq F_{S,j}^{-1}(\gamma)$ help the prior-independent trade reduction mechanism obtain a $(1-\eps)$-approximation to the optimal GFT. While we prove that $O(1/\eps \gamma^2)$ agents suffices, the lower bound we construct requires only $\Omega(1/\eps\gamma)$ agents. Closing this gap is also an interesting direction.

\bibliography{refs.bib}

\newcommand{\etalchar}[1]{$^{#1}$}
\begin{thebibliography}{CBGdK{\etalchar{+}}17}

\bibitem[BCGZ18]{BabaioffCGZ18}
Moshe Babaioff, Yang Cai, Yannai~A Gonczarowski, and Mingfei Zhao.
\newblock The best of both worlds: Asymptotically efficient mechanisms with a guarantee on the expected gains-from-trade.
\newblock In {\em Proceedings of the 2018 ACM Conference on Economics and Computation}, pages 373--373, 2018.

\bibitem[BCWZ17]{BrustleCWZ17}
Johannes Brustle, Yang Cai, Fa~Wu, and Mingfei Zhao.
\newblock Approximating gains from trade in two-sided markets via simple mechanisms.
\newblock In {\em Proceedings of the 2017 ACM Conference on Economics and Computation}, pages 589--590, 2017.

\bibitem[BD21]{BlumrosenD21}
Liad Blumrosen and Shahar Dobzinski.
\newblock (almost) efficient mechanisms for bilateral trading.
\newblock {\em Games and Economic Behavior}, 130:369--383, 2021.

\bibitem[BGG20]{BabaioffGG20}
Moshe Babaioff, Kira Goldner, and Yannai~A Gonczarowski.
\newblock Bulow-{K}lemperer-style results for welfare maximization in two-sided markets.
\newblock In {\em Proceedings of the Fourteenth Annual ACM-SIAM Symposium on Discrete Algorithms}, pages 2452--2471. SIAM, 2020.

\bibitem[BK94]{BulowK94}
Jeremy~I Bulow and Paul~D Klemperer.
\newblock Auctions vs. negotiations, 1994.

\bibitem[BLM13]{boucheron2013concentration}
St{\'e}phane Boucheron, G{\'a}bor Lugosi, and Pascal Massart.
\newblock {\em Concentration inequalities: A nonasymptotic theory of independence}.
\newblock Oxford university press, 2013.

\bibitem[BM16]{BlumrosenM16}
Liad Blumrosen and Yehonatan Mizrahi.
\newblock Approximating gains{-}from{-}trade in bilateral trading.
\newblock In {\em Web and Internet Economics: 12th International Conference, WINE 2016, Montreal, Canada, December 11-14, 2016, Proceedings 12}, pages 400--413. Springer, 2016.

\bibitem[BW19]{BeyhagiW19}
Hedyeh Beyhaghi and S~Matthew Weinberg.
\newblock Optimal (and benchmark-optimal) competition complexity for additive buyers over independent items.
\newblock In {\em Proceedings of the 51st Annual ACM SIGACT Symposium on Theory of Computing}, pages 686--696, 2019.

\bibitem[CBGdK{\etalchar{+}}17]{colini2017fixed}
Riccardo Colini-Baldeschi, Paul Goldberg, Bart de~Keijzer, Stefano Leonardi, and Stefano Turchetta.
\newblock Fixed price approximability of the optimal gain from trade.
\newblock In {\em International Conference on Web and Internet Economics}, pages 146--160. Springer, 2017.

\bibitem[CBGK{\etalchar{+}}20]{colini2020approximately}
Riccardo Colini-Baldeschi, Paul~W Goldberg, Bart~de Keijzer, Stefano Leonardi, Tim Roughgarden, and Stefano Turchetta.
\newblock Approximately efficient two-sided combinatorial auctions.
\newblock {\em ACM Transactions on Economics and Computation (TEAC)}, 8(1):1--29, 2020.

\bibitem[CBKLT16]{Colini-Baldeschi16}
Riccardo Colini-Baldeschi, Bart~de Keijzer, Stefano Leonardi, and Stefano Turchetta.
\newblock Approximately efficient double auctions with strong budget balance.
\newblock In {\em Proceedings of the twenty-seventh annual ACM-SIAM symposium on Discrete algorithms}, pages 1424--1443. SIAM, 2016.

\bibitem[CGMZ21]{CaiGMZ21}
Yang Cai, Kira Goldner, Steven Ma, and Mingfei Zhao.
\newblock On multi-dimensional gains from trade maximization.
\newblock In {\em Proceedings of the 2021 ACM-SIAM Symposium on Discrete Algorithms (SODA)}, pages 1079--1098. SIAM, 2021.

\bibitem[CS21]{CaiS21}
Linda Cai and Raghuvansh~R Saxena.
\newblock 99\% revenue with constant enhanced competition.
\newblock In {\em Proceedings of the 22nd ACM Conference on Economics and Computation}, pages 224--241, 2021.

\bibitem[CW23]{CaiW23}
Yang Cai and Jinzhao Wu.
\newblock On the optimal fixed-price mechanism in bilateral trade.
\newblock {\em arXiv preprint arXiv:2301.05167}, 2023.

\bibitem[DMSW22]{DengMSW22}
Yuan Deng, Jieming Mao, Balasubramanian Sivan, and Kangning Wang.
\newblock Approximately efficient bilateral trade.
\newblock In {\em Proceedings of the 54th Annual ACM SIGACT Symposium on Theory of Computing}, pages 718--721, 2022.

\bibitem[DRS09]{DughmiRS09}
Shaddin Dughmi, Tim Roughgarden, and Mukund Sundararajan.
\newblock Revenue submodularity.
\newblock In {\em Proceedings of the 10th ACM conference on Electronic commerce}, pages 243--252, 2009.

\bibitem[DRT14]{DuttingRT14}
Paul D{\"u}tting, Tim Roughgarden, and Inbal Talgam{-}Cohen.
\newblock Modularity and greed in double auctions.
\newblock In {\em Proceedings of the fifteenth ACM conference on Economics and computation}, pages 241--258, 2014.

\bibitem[EFF{\etalchar{+}}17]{EdenFFTW17}
Alon Eden, Michal Feldman, Ophir Friedler, Inbal Talgam-Cohen, and S~Matthew Weinberg.
\newblock The competition complexity of auctions: A bulow-klemperer result for multi-dimensional bidders.
\newblock In {\em Proceedings of the 2017 ACM Conference on Economics and Computation}, pages 343--343, 2017.

\bibitem[Fei22]{Fei22}
Yumou Fei.
\newblock Improved approximation to first-best gains-from-trade.
\newblock In {\em Web and Internet Economics: 18th International Conference, WINE 2022, Troy, NY, USA, December 12--15, 2022, Proceedings}, pages 204--218. Springer, 2022.

\bibitem[FFR18]{FeldmanFR18}
Michal Feldman, Ophir Friedler, and Aviad Rubinstein.
\newblock 99\% revenue via enhanced competition.
\newblock In {\em Proceedings of the 2018 ACM Conference on Economics and Computation}, pages 443--460, 2018.

\bibitem[FLR19]{FuLR19}
Hu~Fu, Christopher Liaw, and Sikander Randhawa.
\newblock The vickrey auction with a single duplicate bidder approximates the optimal revenue.
\newblock In {\em Proceedings of the 2019 ACM Conference on Economics and Computation}, pages 419--420, 2019.

\bibitem[HR09]{HartlineR09}
Jason~D Hartline and Tim Roughgarden.
\newblock Simple versus optimal mechanisms.
\newblock In {\em Proceedings of the 10th ACM conference on Electronic commerce}, pages 225--234, 2009.

\bibitem[KPV22]{KangPV22}
Zi~Yang Kang, Francisco Pernice, and Jan Vondr{\'a}k.
\newblock Fixed-price approximations in bilateral trade.
\newblock In {\em Proceedings of the 2022 Annual ACM-SIAM Symposium on Discrete Algorithms (SODA)}, pages 2964--2985. SIAM, 2022.

\bibitem[LP18]{LiuP18}
Siqi Liu and Christos-Alexandros Psomas.
\newblock On the competition complexity of dynamic mechanism design.
\newblock In {\em Proceedings of the Twenty-Ninth Annual ACM-SIAM Symposium on Discrete Algorithms}, pages 2008--2025. SIAM, 2018.

\bibitem[McA92]{McAfee92}
R~Preston McAfee.
\newblock A dominant strategy double auction.
\newblock {\em Journal of Economic Theory}, 56(2):434--450, 1992.

\bibitem[MPLW22]{MaoPW22}
Jieming Mao, Renato Paes~Leme, and Kangning Wang.
\newblock Interactive communication in bilateral trade.
\newblock In {\em 13th Innovations in Theoretical Computer Science Conference (ITCS 2022)}. Schloss Dagstuhl-Leibniz-Zentrum f{\"u}r Informatik, 2022.

\bibitem[MS83]{MyersonS83}
Roger~B Myerson and Mark~A Satterthwaite.
\newblock Efficient mechanisms for bilateral trading.
\newblock {\em Journal of Economic Theory}, 29(2):265--281, 1983.

\bibitem[MU05]{mitzenmacher2005probability}
Michael Mitzenmacher and Eli Upfal.
\newblock Probability and computing: Randomized algorithms and probabilistic analysis, 2005.

\bibitem[Mye81]{Myerson81}
Roger~B Myerson.
\newblock Optimal auction design.
\newblock {\em Mathematics of Operations Research}, 6(1):58--73, 1981.

\bibitem[RTCY12]{roughgarden2012supply}
Tim Roughgarden, Inbal Talgam-Cohen, and Qiqi Yan.
\newblock Supply-limiting mechanisms.
\newblock In {\em Proceedings of the 13th ACM Conference on Electronic Commerce}, pages 844--861, 2012.

\bibitem[SHA18a]{SegalHaleviHA18b}
Erel Segal{-}Halevi, Avinatan Hassidim, and Yonatan Aumann.
\newblock Double auctions in markets for multiple kinds of goods.
\newblock In J{\'{e}}r{\^{o}}me Lang, editor, {\em Proceedings of the Twenty-Seventh International Joint Conference on Artificial Intelligence, {IJCAI} 2018}, pages 489--497. ijcai.org, 2018.

\bibitem[SHA18b]{SegalHaleviHA18a}
Erel Segal{-}Halevi, Avinatan Hassidim, and Yonatan Aumann.
\newblock {MUDA:} {A} truthful multi-unit double-auction mechanism.
\newblock In Sheila~A. McIlraith and Kilian~Q. Weinberger, editors, {\em Proceedings of the Thirty-Second {AAAI} Conference on Artificial Intelligence, (AAAI-18)}, pages 1193--1201. {AAAI} Press, 2018.

\bibitem[SS13]{SivanS13}
Balasubramanian Sivan and Vasilis Syrgkanis.
\newblock Vickrey auctions for irregular distributions.
\newblock In {\em Web and Internet Economics: 9th International Conference, WINE 2013, Cambridge, MA, USA, December 11-14, 2013, Proceedings 9}, pages 422--435. Springer, 2013.

\bibitem[Vic61]{Vickrey61}
William Vickrey.
\newblock Counterspeculation, auctions, and competitive sealed tenders.
\newblock {\em The Journal of Finance}, 16(1):8--37, 1961.

\end{thebibliography}
\bibliographystyle{alpha}

\appendix
\section{Basic Facts and Claims}
\label{app:facts}
%\chris{TODO: add a citation.}
% \begin{lemma}[Hoeffding's Inequality]
%     \label{lemma:hoeffding}
%     Let $X_1, \ldots, X_n$ be independent random variables in $[0, 1]$.
%     Let $S_n = \sum_{i=1}^n X_n$.
%     Then
%     \[
%         \prob{S_n - \expect{S_n} \geq t} \leq \exp\left( -\frac{2t^2}{n} \right)
%         \quad\text{and}\quad
%         \prob{S_n - \expect{S_n} \leq -t} \leq \exp\left( -\frac{2t^2}{n} \right).
%     \]
% \end{lemma}

\begin{lemma}[Chernoff Bound~(e.g.~{\cite[Exercise~2.10]{boucheron2013concentration}, \cite[Theorem~4.4, Theorem~4.5]{mitzenmacher2005probability}}]
    \label{lemma:chernoff}
    Let $X_1, \ldots, X_n$ be independent random variables in $[0, R]$.
    Let $S = \sum_{i=1}^n X_i$ and $\mu = \expect{S}$.
    Then for every $\delta \in [0, 1]$,
    \[
        \prob{S - \mu \geq \delta \mu} \leq \exp\left( -\delta^2 \mu / 3R \right)
        \quad\text{and}\quad
        \prob{S - \mu \leq -\delta \mu} \leq \exp\left( -\delta^2 \mu / 3R \right).
    \]
    Moreover, if $\delta \geq 1$ then
    \[
        \prob{S - \mu \geq \delta \mu} \leq \exp(-\delta \mu / 3 R).
    \]
\end{lemma}

\begin{claim}\label{claim:c_ineq}
    If $c \geq 2000$ then $\frac{80 \log(12c)}{c} \leq 1/2$.
    If $c \geq 150$ then $\frac{10}{c} \log(12c) \leq 1/2$.
    If $c \geq 20000$ then $\frac{1280\log(12c)}{c} \leq 0.8$.
\end{claim}
\begin{proof}
    Let $f(c) = \frac{80}{c} \log(12c)$.
    One can check that $f(2000) < 1/2$.
    Moreover, $f'(c) = -\frac{80}{c^2} (\log(12c) - 1) < 0$ since $\log(12c) > \log(e) = 1$. Other statements are similar.
\end{proof}

\begin{claim}
    \label{claim:xexp_lt_exp}
    If $x \geq 4$ then $xe^{-x} \leq e^{-x/2}$.
\end{claim}
\begin{proof}
    The inequality is equivalent to $\log(x) - x \leq -x/2$, which in turn is equivalent to $x/2 - \log(x) \geq 0$.
    It is easy to see that the inequality holds for $x = 4$.
    It holds for $x \geq 4$ since the derivative of $x/2 - \log(x)$ is $1/2 - 1/x \geq 1/4 > 0$.
\end{proof}

\begin{claim}
    \label{claim:one_over_one_minus_x_lt_exp}
    For $0 \leq x \leq 1/4$, $\frac{1}{1-x} \leq 1+2x \leq e^{2x}$.
\end{claim}
\begin{proof}
    The first inequality is standard and holds for all $x \geq 0$.
    We prove only the first inequality.
    Let $f(x) = 1 + 2x - \frac{1}{1-x}$.
    Note that $f(0) = 0$ and $f(1/4) = 3/2 - 4/3 > 0$.
    Hence, it suffices to check that $f$ is convex on $[0, 1/4]$.
    Indeed, $f'(x) = \frac{1}{(1-x)^2}$ and $f''(x) = -\frac{2}{(1-x)^3} < 0$.
\end{proof}

\begin{claim}
    \label{claim:prob_condition_lb}
    Let $N \geq 1, c \geq 1$ be integers.
    Fix sets $I \subseteq [N]$ and $K \subseteq [N] \setminus I$.
    Let $X$ be a uniformly random subset of $[N]$ such that $|X| = c$.
    Then for every $r \geq 0$,
    \[
        \probg{|X \cap I| \geq r}{X \cap K = \emptyset} \geq \prob{|X \cap I| \geq r}.
    \]
\end{claim}
A proof of Claim~\ref{claim:prob_condition_lb} can be found in Appendix~\ref{app:prob_condition_lb}.

\subsection{Proof of Claim~\ref{claim:prob_condition_lb}}
\label{app:prob_condition_lb}
\begin{claim}
    \label{claim:binom_like_concave}
    Let $c, x, y$ be positive integers such that $y > x > c$.
    Let
    \[
        f(t) =
        \frac{\binom{x+1}{c-t}}{\binom{y+1}{c}} - \frac{\binom{x}{c-t}}{\binom{y}{c}} 
    \]
    defined for $t \in \{0, \ldots, c\}$.
    There is some $T \in \{0, \ldots, c-1\}$ such that $f(t) \geq 0$ for $t \leq T$ and $f(t) < 0$ for $t > T$.
\end{claim}
\begin{proof}
    Simplifying, we can write
    \[
        f(t) = \underbrace{\left[
            \frac{x+1}{x+1-c+t} \cdot \frac{y+1-c}{y+1} - 1 
        \right]}_{\eqqcolon g(t)} \cdot
        \frac{\binom{x}{c-t}}{\binom{y}{c}}
    \]
    Notice that $g(t)$ is (strictly) decreasing in $t$.
    Moreover, it is straightforward to show that $g(0) > 0$ and $g(c) < 0$.
    We conclude there is some $T \in \{0, \ldots, c-1\}$ such that $f(t) \geq 0$ for $t \leq T$ and $f(t) < 0$ for $t > T$.
\end{proof}
\begin{claim}
    \label{claim:cdf_pos}
    Let $N, I, p$ be positive integers such that $N \geq \max\{I, p\}$.
    For every integer $k \geq 0$
    \[
        \sum_{t=\max\{0, p-(N-I)\}}^r \frac{\binom{I}{t} \binom{N-I}{p-t}}{\binom{N}{p}} 
        \geq 
        \sum_{t=\max\{0, p-(N-I-k)\}}^r \frac{\binom{I}{t} \binom{N-k-I}{p-t}}{\binom{N-k}{p}}.
    \]
\end{claim}
\begin{proof}
    We prove the claim for $k = 1$; the general version follows by induction.
    Define
    \[
        f(t) = \frac{\binom{N-I}{p-t}}{\binom{N}{p}}
        -
        \frac{\binom{N-1-I}{p-t}}{\binom{N-1}{p}}
    \]
    Let $r_0 = \max\{0, p-(N-I)\}$ and $r_1 = \max\{0, p-(N-I-1)\}$.
    Then the claim is equivalent to
    \[
        g(r) \coloneqq
        \sum_{t=r_0}^r \frac{\binom{I}{t} \binom{N-I}{p-t}}{\binom{N}{p}} 
        -
        \sum_{t=r_1}^r \frac{\binom{I}{t} \binom{N-1-I}{p-t}}{\binom{N-1}{p}}
        = \frac{\binom{I}{r_0} \binom{N-I}{p-r_0}}{\binom{N}{p}} \ind{r_0 \neq r_1, r \geq r_0} +  \sum_{t=r_1}^r \binom{I}{t} \cdot f(t) \geq 0.
    \]
    Note that $g(r_1 - 1) \geq 0$.
    To see this, if $r_1 = r_0$ then $r_1 = r_0 = 0$ so $r_1 - 1 = -1$ and $g(-1)$ corresponds to an empty sum.
    Otherwise, $g(r_1 - 1)$ consists only a single positive term.
    In addition, observe that $g(I) = 0$.
    This is because $\sum_{t=r_0}^{\min\{I, p\}} \binom{I}{t} \binom{N-I}{p-t} = \binom{N}{p}$ since both sides count the number of ways to draw $p$ items from a set of size $N$.
    Similarly, $\sum_{t=r_1}^{\min\{I,p\}} \binom{I}{t} \binom{N-1-I}{p-t} = \binom{N-1}{p}$.
    Claim~\ref{claim:binom_like_concave} implies that there exists $T \in \{0, \ldots, I\}$ such that $g(r)$ is non-decreasing for $r \leq T$ and decreasing for $r > T$.
    Since $g(r_1 - 1) \geq 0$ and $g(I) = 0$ we conclude that $g(r) \geq 0$ for all $r$.
\end{proof}
%\chris{Change to prevproof environment.}
\begin{proof}[Proof of Claim~\ref{claim:prob_condition_lb}]
    If $r = 0$ or $r > \min\{|I|, c\}$ then the claim is trivial so we assume that $0 < r \leq \min\{|I|, c\}$.
    Let $k = |K|$.
    To choose a set $X$ such that $|X \cap I| = t$ and $X \cap K = \emptyset$ we can first choose $t$ elements from $I$ to add to $X$
    and then choose $c-t$ elements from the remaining $N - k - |I|$ elements in $[N]$.
    Notice that $t \geq \max\{0, c - (N - k - |I|)\} \eqqcolon r_k$ (otherwise it is impossible to choose $c$ elements).
    Thus,
    \begin{equation}
        \label{eqn:XcapKempty}
        \probg{|X \cap I| < r}{X \cap K = \emptyset}
        =
        \sum_{t=r_k}^{r-1} \frac{\binom{|I|}{t} \binom{N-k-|I|}{c-t}}{\binom{N-k}{c}}.
    \end{equation}
    Similarly, letting $r_0 = \max\{0, c - (N-|I|)\}$, we have
    \begin{equation}
        \label{eqn:noXcapK}
        \prob{|X \cap I| < r}
        =
        \sum_{t=r_0}^{r-1} \frac{\binom{|I|}{t} \binom{N-|I|}{c-t}}{\binom{N}{c}}.
    \end{equation}
    By Claim~\ref{claim:cdf_pos}, we conclude that Eq.~\eqref{eqn:XcapKempty} is upper bounded by Eq.~\eqref{eqn:noXcapK} which proves the claim by taking complements.
\end{proof}
\section{Comparison with per-unit GFT}
\label{app:per_unit_gft}
In this section, we consider the per-unit GFT which is defined as the total GFT divided by the number of sellers.
When there are $m$ buyers and $1$ seller, Babaioff, Goldner, and Gonczarowski \cite[Theorem~5.1]{BabaioffGG20} give an example where if one is restricted to recruit only buyers than $\Omega(\log m)$ buyers are necessary for the per-unit GFT in the augmented market using a prior-independent mechanism to exceed the first-best GFT in the original market.
A natural question is whether or not it is possible to recruit $O(1)$ sellers and $o(\log m)$ buyers so that the total number of additional agents is $o(\log m)$.

In this section, we provide an example which shows that, $\Omega(\log m)$ additional agents is necessary for any prior-independent mechanism to achieve at least the same per-unit GFT as the optimal allocation without augmentation.
More specifically, we describe an instance with $m$ buyers and $1$ seller where if we add $s$ sellers (for $1 \leq s < \sqrt{m}$) then we require $b \geq \Omega(s \log m)$ buyers just for the optimal per-unit GFT in the augmented market to exceed the optimal per-unit GFT in the original market.

The instance we consider is the following which is identical to the instance that appears in \cite{BabaioffGG20}.
The buyer distribution, $F_B$, is as follows. With probability $0.5$, the buyer value is $2$ and otherwise, the buyer value is $0$.
For the seller distribution, $F_S$, we assume that the seller value is $1$ with probability $0.5$ and otherwise, it is equal to $0$.
Clearly, the buyer distribution FSD the seller distribution.

We assume that $m$ is the number of original buyers and $n = 1$ is the number of original sellers.
Let $X_B$ be the number of original buyers with value $2$.
If the seller has value $0$ then the optimal per-unit GFT is given $2\min(X_B, 1)$ and if the seller has value $1$ then the optimal (per-unit) GFT is given by $\min(X_B, 1)$.
Since the buyer and seller values are independent, the expected per-unit GFT is
\begin{equation}
    \label{eqn:per_unit_gft_n_eq_1}
    \OPT(m, 1) = 1.5 \cdot \expect{\min(X_B, 1)} = 1.5 \cdot \prob{X_B \geq 1} = 1.5 \cdot (1 - 2^{-m}).
\end{equation}

Now, let us assume we have $m+b$ buyers and $1+s$ sellers.
Note that for $k \leq s$, we have
\[
    \expectg{\OPT(m+b, 1+s)}{X_B = k}
    = k + \expect{\min(X_S, k)}
    \leq s + \expect{\min(X_S, b)}
    = 1.5s + 0.5 - 2^{-(s+1)},
\]
where the last equality used Claim~\ref{claim:binom_min_X_s}.
We also have that $\expectg{\OPT(m+b, 1+s)}{X_B \geq 1+s} = 1.5(1+s)$.
Thus, we have
\begin{align*}
    \expect{\OPT(m+b, 1+s)}
    & = \expectg{\OPT(m+b, 1+s)}{X_B \leq s} \prob{X_B \leq s} \\
    & + \expectg{\OPT(m+b, 1+s)}{X_B \geq 1+s} \prob{X_B \geq 1+s} \\
    & \leq \left(1.5(s + 1) - 1 - 2^{-(s+1)} \right) \cdot \prob{X_B \leq s} \\
    & + 1.5(s+1) \cdot \left( 1 - \prob{X_B \leq s} \right) \\
    & = 1.5(s+1) - \left(1 + 2^{-(s+1)}\right) \cdot \prob{X_B \leq s} \\
    & \leq 1.5(s+1) - \prob{X_B \leq s}.
\end{align*}
Thus, the per-unit GFT satisfies
\begin{equation}
    \label{eqn:per_unit_gft_n_eq_s_plus_1}
    \frac{\expect{\OPT(m+b, 1+s)}}{1+s} \leq 1.5 - \frac{\prob{X_B \leq s}}{1+s}.
\end{equation}
Comparing Eq.~\eqref{eqn:per_unit_gft_n_eq_1} and comparing Eq.~\eqref{eqn:per_unit_gft_n_eq_s_plus_1},
we have that a sufficient condition for the per-unit GFT with $m$ buyers and $1$ seller to be strictly larger than the per-unit GFT with $m+b$ buyers and $1+s$ sellers is if
\[
    2^{-m} < \frac{\prob{X_B \leq s}}{1+s}.
\]
Note that we have $\prob{X_B \leq s} \geq \prob{X_B = s} = \binom{m+b}{s}2^{-(m+b)} \geq \left(\frac{m+b}{s}\right)^s 2^{-(m+b)} > \left( \frac{m}{s} \right)^s 2^{-(m+b)}$.
Thus, a sufficient condition for the above inequality to hold is
\[
    2^{-m} < \frac{1}{s+1} \left( \frac{m}{s} \right)^s 2^{-(m+b)},
\]
which, after rearranging, is equivalent to
\[
    2^b < \frac{1}{s+1} \left( \frac{m}{s} \right)^s.
\]
Thus, we conclude that $b \geq s \log_2(m / s) - \log_2(s+1)$ is necessary for the per-unit GFT in the augmented market to be at least the per-unit GFT in the original market.
Note that this last bound implies that $\Omega(\log m)$ additional agents are required for the per-unit GFT in the augmented market to be at least the per-unit GFT in the original market.
If $s = 1$ then we require $b \geq \Omega(\log m)$.
and if $2 \leq s \leq \sqrt{m} - 1$ then the inequality requires $b \geq s \log_2(\sqrt{m}) - \log_2(\sqrt{m}) = \frac{s-1}{2} \log_2 m$.

\begin{claim}
    \label{claim:binom_min_X_s}
    Let $s \geq 1$ be an integer and suppose $X \sim \Binom(0.5, s+1)$.
    Then $\expect{\min(X, s)} = 0.5(s+1) - 2^{-(s+1)}$.
\end{claim}
\begin{proof}
    Note that we can write $\expect{\min(X, s)} = \expect{\min(X, s+1)} - \prob{X = s+1} = 0.5(s+1) - 2^{-(s+1)}$,
    where in the second equality, we used that $\min(X, s+1) = X$ and $\expect{X} = 0.5(s+1)$.
\end{proof}
\section{Missing Proofs from Section~\ref{sec:str}}
\label{app:str}

We prove \Cref{thm:main_str} in this section. We first provide some notations used throughout this section. Let $i^{\pi}_1 \leq \ldots \leq i^{\pi}_{m+c}$ be the indices of all buyers, in an decreasing order of their quantiles and thus their values. 
Let $j^{\pi}_1 \geq \ldots \geq j^{\pi}_{m+c}$ be the indices of all sellers (in an increasing order of their quantiles and thus their values). 
%Similarly, let $i^{\pi, O}_1 \leq \ldots \leq i^{\pi, O}_{m}$ and $i^{\pi, N}_1 \leq \ldots \leq i^{\pi, N}_{m}$
Similarly, let $\iold_1 \leq \ldots \leq \iold_{m}$ be the indices of all old buyers, $\inew_1 \leq \ldots \leq \inew_{c}$ be the indices of all new buyers. Let let $\jnew_1 \geq \ldots \geq \jold_{n}$ be the indices of all old sellers, $\jnew_1 \geq \ldots \geq \jnew_{c}$ be the indices of all new sellers.

\subsection{Proof of Lemma~\ref{lemma:E1}}
\label{subsec:proof_E1}

\begin{claim}
    \label{claim:all_I2J2_trades}
    In the original first-best matching, every original buyer in $I_1\cup I_2$ is matched and every original seller in $J_1\cup J_2$ is matched. Similarly, in the augmented first-best matching, every buyer in $I_1\cup I_2$ is matched and every seller in $J_1\cup J_2$ is matched.
\end{claim}
\begin{proof}
    For the first statement, we prove only that every original buyer in $I_1\cup I_2$ is matched. The proof for the sellers is analogous.
    
%    Let $\iold_{b, 1} < \ldots < \iold_{b, m}$ be the indices of the $m$ original buyers
%    and $\jold_{s, 1} > \ldots > \jold_{s, n}$ be the indices of the $m$ original sellers.
%    Similarly, let $\inew_{b, 1} < \ldots < \inew_{b, c}$ be the indices of the $c$ new buyers
%    and $\jnew_{s,1} > \ldots > \jnew_{s, c}$ be the indices of the $c$ new sellers.
%    We also let $i_{b, 1} < \ldots < i_{b, m+c}$ be the indices of all $m+c$ buyers
%    and $j_{s, 1} < \ldots < j_{s, n+c}$ be the indices of all $m+c$ sellers.
    
    Let $k = |(I_1 \cup I_2) \cap \Bold^{\pi}|$ be the number of original buyers in $I_1 \cup I_2$.
    Note that $k \leq 2\cdot \lceil n / 10 \rceil$.
    We prove that there are at least $k$ original sellers outside of $I_1 \cup I_2$, i.e.~$|\Sold^{\pi} \cap ([N] \setminus (I_1 \cup I_2))| \geq k$.
    Indeed, we have
    \begin{align*}
        |\Sold^{\pi} \cap ([N] \setminus (I_1 \cup I_2))|
        & =
        \underbrace{|\Sold^{\pi} \cap [N]|}_{=n} - |\Sold^{\pi} \cap (I_1 \cup I_2)|\\
        &\geq n - 2\left\lceil \frac{n}{10} \right\rceil
        \geq \frac{8n}{10} - 2
        \geq \frac{2n}{10} + 2 \\
        & \geq 2\left\lceil \frac{n}{10} \right\rceil \geq k,
    \end{align*}
    where in the third inequality we use $n \geq 20$.
    Recall that $\iold_k$ is the index of the $k$-th highest original buyer and $\jold_k$ is the index of the $k$-th lowest original seller. The above argument immediately implies that $q(\iold_k)\geq q(\jold_k)$ and thus $b(q(\iold_k))\geq s(q(\jold_k))$ since $F_B$ FSD $F_S$. Thus there are at least $k$ trades in the original first-best matching, which implies that every original buyer in $I_1\cup I_2$ is matched.
    
    The second statement follows from a similar argument. Let $k = |(I_1 \cup I_2) \cap (\Bold^{\pi} \cup \Bnew^{\pi})|$ be the number of buyers in $I_1 \cup I_2$. Note that $k \leq 2\cdot \lceil n / 10 \rceil$.
    We prove that there are at least $k$ sellers outside of $I_1 \cup I_2$, i.e.~$|(\Sold^{\pi}\cup \Snew^{\pi}) \cap ([N] \setminus (I_1 \cup I_2))| \geq k$.
    Indeed, we have
    \begin{align*}
        |(\Sold^{\pi}\cup \Snew^{\pi}) \cap ([N] \setminus (I_1 \cup I_2))|
        & =
        \underbrace{|(\Sold^{\pi}\cup \Snew^{\pi}) \cap [N]|}_{=n+c} - |(\Sold^{\pi}\cup \Snew^{\pi}) \cap (I_1 \cup I_2)|\\
        &\geq n + c - 2\left\lceil \frac{n}{10} \right\rceil
        \geq \frac{8n}{10} + c - 2
        \geq \frac{2n}{10} + 2 \\
        & \geq 2\left\lceil \frac{n}{10} \right\rceil \geq k,
    \end{align*}
    where in the third inequality we use $n \geq 20$.
    Since $F_B$ FSD $F_S$, every buyer in $I_1\cup I_2$ has value no less than the cost of every seller outside of $I_1\cup I_2$. Thus there are at least $k$ trades in the augment first-best matching, which implies that every buyer in $I_1\cup I_2$ is matched.
\end{proof}
% \chris{Currently, we use $i_b$, $j_s$ to denote buyer and seller indices. We can drop the $b$ and $s$. WDYT?}
% \mingfei{Sounds good.}

\begin{proof}[Proof of Lemma~\ref{lemma:E1}]
Let $\pi$ be any assignment in event $\cE_1$. Let $T$ be the number of trades in the original first-best matching. We prove in the following claim that there are at least $T+2$ trades in the augmented first-best matching. An immediate consequence of this is that STR must have at least $T + 1$ trades.

\begin{claim}\label{claim:two_more_trades}
    \label{claim:T_plus_2_trades_v2}
    Recall that $\inew_{2}$ is the index of the second-highest new buyer and $\jnew_{2}$ is the index of the second-lowest new seller. 
    Then 
    $b(q(\inew_{2}))\geq b(q(i_{T+2})) \geq s(q(j_{T+2}))\geq s(q(\jnew_{2}))$. Thus there are at least $T+2$ trades in the augmented first-best matching.
\end{claim}
\begin{proof}
    Let $i'\in I_2 \cap \Bold$ be the index of any original buyer in $I_2$ (by definition of $\cE_1$ there is at least one). By Claim~\ref{claim:all_I2J2_trades}, $i'$ is matched in the original first-best matching and thus $q(i') \geq q(\iold_{T})$. By the property of $\cE_1$ that $|I_1 \cap \Bnew| \geq 2$, we have $q(\inew_{2}) \geq q(i')\geq q(\iold_{T})$ since $\inew_2\in I_1$ while $i'\in I_2$.
    Therefore, $q(\inew_2)\geq q(i_{T+2}) \geq q(\iold_{T})$ as both the highest and second-highest new buyer have quantile no less than $q(\iold_{T})$. A similar argument shows that $q(\jnew_2)\leq q(j_{T+2}) \leq q(\jold_{T})$.
    We conclude that $b(q(\inew_{2}))\geq b(q(i_{T+2})) \geq b(q(\iold_{T})) \geq s(q(\jold_{T})) \geq s(q(j_{T+2}))\geq s(q(\jnew_{2}))$.
    Here the third inequality follows from the fact that there are $T$ trades in the original first-best matching.
\end{proof}

\Cref{claim:two_more_trades} shows that STR trades the $T+1$ highest buyers and the $T+1$ lowest sellers (it may trade more).
Therefore, $\STR(\Bq, \pi)\geq \sum_{t=1}^{T+1}[b(q(i_t))-s(q(j_t))]$. On the other hand, OPT trades the $T$ highest original buyers with the $T$ lowest original sellers and thus, $\OPT(\Bq, \pi)=\sum_{t=1}^{T}[b(q(\iold_t))-s(q(\jold_t))]$.

We claim that $\sum_{t=1}^{T+1}b(q(i_t))-\sum_{t=1}^{T}b(q(\iold_t))\geq b(q(\inew_1))$. This is because by \Cref{claim:two_more_trades}, buyers $\inew_1$ and $\inew_2$ are among the top $T+2$ highest-value buyers. So $\inew_1$ must be in the top $T+1$ highest-value buyers, which is contributed in the first term.
{Note that $\{i_1, \ldots, i_{T+1}\} \setminus \{\inew_1\}$ correspond to the $T$ highest value buyers excluding $\inew_1$ and $\{\iold_1 \ldots, \iold_T\}$ correspond to the top $T$ highest value \emph{original} buyers.
Thus, we conclude that $\sum_{t=1}^{T+1} b(q(i_t)) - b(q(\inew_1)) \geq \sum_{t=1}^T b(q(\iold_t))$.}
% $\sum_{t=1}^{T+1}b(q(i_t))-b(q(\inew_1))$ is the sum of values for the top $T$ buyers excluding $\inew_1$ (a new buyer), which is at least the sum of values for the top $T$ old buyers.
By a similar argument, we have $\sum_{t=1}^{T+1}s(q(j_t))-\sum_{t=1}^{T}s(q(\jold_t))\leq s(q(\jnew_1))$.
Thus
$$\STR(\Bq, \pi)-\OPT(\Bq, \pi)\geq b(q(\inew_1))-s(q(\jnew_1)) \geq 0$$

It remains to lower bound the expected difference between $\STR(\Bq, \pi)$ and $\OPT(\Bq, \pi)$ conditioned on the event $\cE_1$. 
%For any integer $k, \ell\in \{2, \ldots, p=\lceil n/10\rceil\}$, define event $\cF_{k, \ell}$ as the set of assignments such that there are $k$ (and $\ell$) new buyers (and sellers) in $I_1$. 
From the above inequality, $\STR(\Bq, \pi)-\OPT(\Bq, \pi)$ is lower bounded by the value of the highest new buyer subtracting the cost of the lowest new seller. We need the following definition.

\begin{definition}\label{def:swappable}
For any event $\cE$ over an assignment $\pi$, $\cE$ is \emph{swappable} in a set $S$ if: For every $\pi\in \cE$, the assignment $\pi'$ obtained by swapping the label for any two indices in $S$ is also in $\cE$. %Similarly, an event $\cE$ is \emph{seller-swappable} in $J_1$ if for every $\pi\in \cE$, the assignment $\pi'$ obtained by swapping the label for any two sellers in $J_1$ is also in $\cE$.
In other words, for every $\pi \in \cE$ and every $i', i'' \in S$, if $\pi'(i') = \pi(i'')$, $\pi'(i'') = \pi(i')$, and $\pi'(i) = \pi(i)$ for $i \notin \{i', i''\}$ (clearly $\pi'$ is also a valid assignment), then $\pi' \in \cE$.
\end{definition}

\begin{lemma}\label{lem:E1_swappable}
$\cE_1$ is swappable in $I_1$ and it is swappable in $J_1$.
\end{lemma}

\begin{proof}
The lemma directly follows from the fact that swapping the label for any two indices in $I_1$ (or $J_1$) will not change the value of $|I_1 \cap \Bnew^{\pi}|$, $|I_2 \cap \Bold^{\pi}|$, $|J_1 \cap \Snew^{\pi}|$, $|J_2 \cap \Sold^{\pi}|$.
\end{proof}

Consider the following process that generates a random assignment $\pi$ from $\cE_1$: 
%that obtains a uniformly random assignment $\pi$ from $\cE_1$:
% \chris{I really like this. :)}
\begin{enumerate}
    \item Choose an index $i$ uniformly at random from $I_1$ and assign it to the ``New Buyer'' label. Choose an index $j$ uniformly at random from $J_1$ and assign it to the ``New Seller'' label.
    %\item Randomly choose an index from the rest indices in $I_1$ and assign it to the ``New Buyer'' label. Randomly choose an index from the rest indices in $J_1$ and assign it to the ``New Seller'' label.
    %\item Randomly choose an index from $I_2$ and assign it to the ``Old Buyer'' label. 
    %Randomly choose an index from $J_2$ and assign it to the ``Old Seller'' label.     
    %\item Randomly assign the remaining indices.
    \item Denote $\Pi_{i, j}$ the set of valid assignments
    in set $\cE_1$ such that $i$ is assigned to the ``New Buyer'' label and $j$ is assigned to the ``New Seller'' label. Draw an assignment $\pi$ uniformly at random from $\Pi_{i, j}$ and assign the indices accordingly.    
\end{enumerate}

By \Cref{lem:E1_swappable}, we have that $|\Pi_{i', j'}|=|\Pi_{i'', j''}|$ for any indices $i',i''\in I_1, j', j''\in J_1$: For any assignment in $\Pi_{i', j'}$, we can swap the label between indices $i', i''$ and swap between $j, j''$. This generates an assignment in $\Pi_{i'', j''}$ and vice versa. Moreover, for any valid assignment $\pi$, the number of ``New Buyer'' (or ``New Seller'') labels is $c$.
Hence, for every valid assignment $\pi$, $|\{ i \in I_1, j \in J_1 \,:\, \pi \in \Pi_{i,j}\}| = c^2$.
% So $\pi$ is in $c^2$ (a fixed number) different $\Pi_{i', j'}$s.
Thus the above random process chooses the assignment $\pi$ uniformly at random from $\cE_1$.

For any realization of the above process, the value of the highest new buyer is at least $b(q_{i})$ and the cost of the lowest new seller is at most $s(q_{j})$. Thus the difference is at least $b(q_{i})-s(q_{j})$.   
Taking expectation over the random process, we have
$$\bE_{\pi}[\STR(\Bq, \pi) - \OPT(\Bq, \pi)|\cE_1]\geq \expects{i, j}{b(q_{i}) - s(q_{j})}$$ where $i \sim I_1, j \sim J_1$ uniformly at random according to Step 1 of the process.
\end{proof}

\subsection{Proof of Lemma~\ref{lemma:notE2_STR_gt_OPT}}
\label{subsec:notE2_STR_gt_OPT}
\begin{proof}[Proof of \Cref{lemma:notE2_STR_gt_OPT}]
    We know from Lemma~\ref{lemma:E1} that on the event $\cE_1$, we have $\STR(\Bq, \pi) \geq \OPT(\Bq, \pi)$.
    Hence it suffices to show the inequality on the event $\cE'=\{\pi\in \Pi_{n,m,c}\mid\Snew^{\pi} \cap [2n+2c]\not=\Snew^{\pi}\}$. %Let $\OPT' = \OPT(m+c, n+c)$ be the optimal matching even with the augmented buyers and sellers.
    Let $\OPT'$ be the first-best matching in the augmented market. For any $\pi\in \cE'$, we consider two cases based on the number of trades in $\OPT'$ compared with $\OPT$. Suppose $\OPT$ has $T$ trades. Note that the number of trades in $\OPT'$ is least $T$. 
    \paragraph{Case 1: $\OPT'$ has at least $T+1$ trades.}
    %Suppose $\OPT$ has $T$ trades. Then $\OPT'$ has at least $T+1$ trades.
    %In which case, 
    Now in $\STR$ the top $T$ (original and new) buyers and bottom $T$ (original and new) sellers trade.
    The GFT from this is larger than the GFT from $\OPT$ which trades the top $T$ original buyers and the bottom $T$ original sellers.

    \paragraph{Case 2: $\OPT'$ also has $T$ trades.}
    In this case, our goal is to show that $\STR$ has the exact same $T$ trades as $\OPT'$. Thus the GFT of $\STR$ is the same as the GFT of $\OPT'$, which is at least the GFT of $\OPT$.   
    %More precisely, suppose that $T \leq n$ is the number of trades in $\OPT'$ (and $\OPT$).
    By definition of STR, it suffices to show that $b(q(i_{T})) \geq s(q(j_{T+1}))$, i.e. the $T$-th highest buyer value is least the $(T+1)$-th lowest seller cost.
    %\begin{equation}
    %    \label{eqn:bT_gt_qT+1}
    %    b(q(i_{T})) \geq s(q(j_{T+1})) \geq s(q(j_{s,T})).
    %\end{equation}
    %The second inequality is true by definition so we only prove the first inequality.

    First, we claim that $i_{n+c} \leq 2n+2c$. Indeed, there are $n+c$ sellers in the augmented market. Thus, the $(n+c)$-th highest value buyer must have index at most $2n+2c$.

    Next, recall that $\jold_{T}$ is the index of the $T$-th lowest value \emph{original} seller.
    We claim that $\jold_{T} \leq 2n+2c$.
    For the sake of contradiction, suppose $\jold_{T} > 2n+2c$. Recall that $\jnew_{1}$ is the index of the lowest-value new seller.
    Since $\pi\in \cE'$, we have $\Snew \cap [2n+2c] \neq \Snew$ and thus, $\jnew_1 > 2n+2c$.
    In particular, $j_{T+1} \geq \min\{\jold_{T}, \jnew_{1}\} > 2n+2c \geq i_{n+c} \geq i_{T+1}$.
    Thus $q(i_{T+1})>q(j_{T+1})$ and $b(q(i_{T+1}))>s(q(j_{T+1}))$. This implies that $\OPT'$ has at least $T + 1$ trades, a contradiction.
    
    To finish the proof, we have
    \[
        b(q(i_{T}))
        \geq b(q(\iold_{T}))
        \geq s(q(\jold_{T}))
        \geq s(q(j_{T+1})).
    \]
    The first inequality uses $i_T\leq \iold_T$. The second inequality follows from the fact that $\OPT$ has $T$ trades. The last inequality holds because $\jold_{T}\leq 2n+2c<\jnew_{1}$ and thus, $j_{T+1} \geq \min\{\jold_{T}, \jnew_{1}\} = \jold_{T}$.
    %\mingfei{A bit confused here. We proved $\jold_T<\jnew_1$ so their ``min'' is $\jold_T$.}
    %\chris{Oops, I think I got confused.}
    We conclude that $\STR$ has the exact same $T$ trades as $\OPT'$.
\end{proof}

\subsection{Proof of Lemma~\ref{lemma: bound_loss_E2}}
\label{subsec:proof_bound_loss_E2}
% \begin{prevproof}{Lemma}{lemma: bound_loss_E2}
\begin{proof}[Proof of Lemma~\ref{lemma: bound_loss_E2}]
% \begin{proof}
For every $(\Bq, \pi)$, let $\OPT'(\Bq, \pi)$ be the GFT of the first-best matching in the augmented market. We clearly have $\OPT(\Bq, \pi)\leq \OPT'(\Bq, \pi)$. For each $(\Bq, \pi)$, we let $b^*(\Bq,\pi)$ (resp.~$s^*(\Bq, \pi)$) denote the lowest value among buyers (resp.~the largest value among sellers) traded in the augmented first-best matching. {Let $\cF$ be the event that there is no trade in the augmented first-best matching and define $b^*(\Bq,\pi)= 0$ and $s^*(\Bq, \pi) = 0$ if there is no trade.} Then by definition of the $\STR$ mechanism, $\OPT(\Bq, \pi) - \STR(\Bq, \pi)\leq \OPT'(\Bq, \pi)-\STR(\Bq, \pi)\leq (b^*(\Bq, \pi)-s^*(\Bq,\pi)) {\cdot \ind{\pi \in \cF}}$.
% \chris{I suggest defining $b^*(\Bq, \pi) = -\infty$ and $s^*(\Bq, \pi) = \infty$ if there is no trade because (i) to avoid requiring non-negativity in the proof and (ii) I think $\cE'''$ might have an issue if $s^*(\Bq, \pi) = 0$ since we cannot lower bound it by something that is uniformly random.}

%To bound the expected difference $b^*(\Bq, \pi)-s^*(\Bq,\pi)$ conditioned on $\cE_2$, we need the following definition.

% \chrisnote{We will consider three cases and show that in each case $\expects{i, j}{b(q(i)) - s(q(j))}$, where $i \sim I_1$ and $j \sim J_1$ uniformly at random, is an upper bound on in the expectation of $\OPT(\Bq, \pi) - \STR(\Bq, \pi)$.
% We will mainly focus on bounding $b^*(\Bq, \pi)$ since the argument for $s^*(\Bq, \pi)$ is analogous.}
%\chris{Added $\ind{\cF}$.}
We will show that $\expectg{b^*(\Bq, \pi)\ind{\pi \in \cF}}{\cE_2} \leq \expects{i}{b(q(i))}$ where $i \sim I_1$ uniformly at random.
A similar argument shows that $\expectg{s^*(\Bq, \pi) \ind{\pi \in \cF}}{\cE_2} \geq \expects{j}{s(q(j))}$ where $j \sim J_1$ uniformly at random.
To do so, we consider three cases: (i) where at least one new buyer is in $I_1$, (ii) where no new buyers are in $I_1$ but at least one original buyer is in $I_1$, and (iii) where no buyers (original or new) are in $I_1$.
In each of these case, we prove that $\expects{i}{b(q(i))}$ is an upper bound on the value of lowest value traded buyer in the augmented market, in expectation.

% We consider $\expectg{b^*(\Bq, \pi)}{\cE_2}$, the expected lowest value among buyers traded in the augment first-best matching, conditioned on $\cE_2$.
\paragraph{Case 1: $|I_1 \cap \Bnew^{\pi}| \geq 1$.}
Let $\cE'$ be the event that $|I_1 \cap \Bnew^{\pi}| \geq 1$, i.e. at least one new buyer is in $I_1$.
% Let $\cE'$ be the event such that at least one new buyer is in $I_1$, i.e. $|I_1\cap \Bnew^{\pi}|\geq 1$.
The following lemma is similar to \Cref{lem:E1_swappable}, which immediately follows from the definition of $\cE'$ and $\cE_2$.

\begin{lemma}\label{lem:swappable}
$\cE'\cap\cE_2$ is swappable (see \Cref{def:swappable}) in $I_1$. Moreover, it is swappable in $J_1$.
\end{lemma}

\begin{proof}
Let $\pi$ be any assignment in $\cE'\cap\cE_2=\cE'\cap \neg \cE_1 \cap \{\hat{\pi}\in \Pi_{n,m,c}\mid\Snew^{\hat{\pi}} \subseteq [2n+2c]\}$ and $\pi'$ be the assignment obtained by swapping any two labels in $I_1$ (or $J_1$). We notice that swapping the label for any two indices in $I_1$ (or $J_1$) will not change the value of $|I_1 \cap \Bnew^{\pi}|$, $|I_2 \cap \Bold^{\pi}|$, $|J_1 \cap \Snew^{\pi}|$, $|J_2 \cap \Sold^{\pi}|$. Thus the new assignment $\pi'$ is also in $\cE'\cap\neg\cE_1$. Moreover, since $I_1\subseteq [2n+2c]$, $\Snew^{\pi} \subseteq [2n+2c]$ implies that $\Snew^{\pi'} \subseteq [2n+2c]$. Thus $\pi\in \{\hat{\pi}\in \Pi_{n,m,c}\mid\Snew^{\hat{\pi}} \subseteq [2n+2c]\}$.
\end{proof}

% Back to the proof of \Cref{lemma: bound_loss_E2}.
Consider the following random process of choosing an assignment $\pi$:

\begin{enumerate}
    \item Choose an index $i$ uniformly at random from $I_1$ and assign it to the ``New Buyer'' label.
    \item Denote $\Pi_{i}$ the set of valid assignments %$\pi_{-i} \colon [N]\backslash \{i\} \to \{\BO, \BN, \SO, \SN\}$ 
    in set $\cE'\cap \cE_2$ such that $i$ is assigned to the ``New Buyer'' label. Draw an assignment $\pi$ uniformly at random from $\Pi_{i}$ and assign the indices accordingly.
\end{enumerate}

By \Cref{lem:swappable}, we have that $|\Pi_{i'}|=|\Pi_{i''}|$ for any indices $i',i''\in I_1$: For any assignment in $\Pi_{i'}$, we can swap the label for index $i'$ and index $i''$ and generate an assignment in $\Pi_{i''}$ and vice versa. Moreover, for any valid assignment $\pi$, the number of ``New Buyer'' labels is $c$. So $\pi$ is in $c$ (a fixed number) different $\Pi_{i'}$s. Thus the above random process chooses the assignment $\pi$ uniformly random from $\cE'\cap \cE_2$.

For any realization of the above process, we notice that by \Cref{claim:all_I2J2_trades}, the new buyer with index $i$ trades in the augmented first-best matching and therefore $\ind{\pi \in \cF} = 1$. Thus $b^*(\Bq, \pi)$, the lowest value among buyers traded in the augmented first-best matching, is upper bounded by $b(q(i))$. Thus $\bE_{\pi}[b^*(\Bq, \pi) \ind{\pi \in \cF} |\cE'\cap\cE_2]\leq \bE_i[b(q(i))]$, where $i$ draws from $I_1$ uniformly at random.

\paragraph{Case 2: $|I_1 \cap \Bnew^{\pi}| = 0$ and $|I_1 \cap \Bold^{\pi}| \geq 1$.}
Next, let $\cE''$ be the event such that no new buyer is in $I_1$ and at least one old buyer is in $I_1$, i.e. $|I_1\cap \Bold^{\pi}|\geq 1$ and $|I_1\cap \Bnew^{\pi}|=0$. One can easily verify that $\cE''\cap \cE_2$ is also swappable in $I_1$. And using a similar argument (by assigning index $i$ to the ``Old Buyer'' label in the random process), we have $\bE_{\pi}[b^*(\Bq, \pi) \ind{\pi \in \cF}|\cE''\cap\cE_2]\leq \bE_i[b(q(i))]$, where $i$ draws from $I_1$ uniformly at random.

\paragraph{Case 3: $|I_1 \cap \Bnew^{\pi}| = 0$ and $|I_1 \cap \Bold^{\pi}| = 0$.}
Finally, let $\cE'''=\neg(\cE'\cup \cE'')$ be the event such that no buyer is in $I_1$.
Then for any $\pi\in \cE'''$, $b^*(\Bq, \pi) \cdot \ind{\pi \in \cF} \leq b(q(\lceil n/10\rceil)) \cdot \ind{\pi \in \cF}$.
To see this, note that if there is no trade (i.e.~$\pi \notin \cF$) then both sides are equal to $0$.
On the other hand, if there is a trade (i.e.~$\pi \in \cF$) then all buyers have at most $b(q(\lceil n / 10 \rceil))$ and thus, so does $b^*(\Bq, \pi)$.
Thus $\bE_{\pi}[b^*(\Bq, \pi) \cdot \ind{\pi \in \cF} |\cE'''\cap\cE_2]\leq b(q(\lceil n/10\rceil)) \prob{\cF} \leq \bE_i[b(q(i))]$, where $i$ draws from $I_1$ uniformly at random.

Summarizing the three inequalities above, we have that $\bE_{\pi}[b^*(\Bq, \pi) \ind{\pi \in \cF}|\cE_2]\leq \bE_i[b(q(i))]$, where $i$ draws from $I_1$ uniformly at random. An analogous argument gives that $\bE_{\pi}[s^*(\Bq, \pi) \ind{\pi \in \cF}|\cE_2]\geq \bE_j[s(q(j))]$, where $j$ draws from $J_1$ uniformly at random. Therefore,
$$\bE_{\pi}[\OPT(\Bq, \pi) - \STR(\Bq, \pi)|\cE_2]\leq \bE_{\pi}[(b^*(\Bq, \pi) - s^*(\Bq, \pi)) \cdot \ind{\pi \in \cF}|\cE_2]\leq  \expects{i, j}{b(q(i)) - s(q(j))},$$
where $i \sim I_1, j \sim J_1$ uniformly at random.
\end{proof}
% \end{prevproof}

\subsection{Proof of Lemma~\ref{lemma:probE1_gt_probE2}}
\label{subsec:proof_probE1_gt_probE2}
% %Before we prove Lemma~\ref{lemma:probE1_gt_probE2}, we give some high-level intuition.
% Recall that the ``good event'' $\cE_1$ as
% \[
%     \cE_1 =
%     \{I_1 \cap \Bnew^{\pi}| \geq 2\} \cap
%     \{I_2 \cap \Bold^{\pi}| \geq 1\} \cap
%     \{J_1 \cap \Snew^{\pi}| \geq 2\} \cap
%     \{J_1 \cap \Sold^{\pi}| \geq 1\}.
% \]
% More formally, we have the following two lemmas which together prove %that when %$m, n, c$ are sufficiently large,
% %we always have $\prob{\cE_1} \geq \prob{\cE_2}$ and completes the proof of Theorem~\ref{thm:str_main}.
% \Cref{lemma:probE1_gt_probE2}.

{To prove \Cref{lemma:probE1_gt_probE2}, we consider two cases depending on whether $n = \Omega(m)$ or $n = O(m)$. This is formalized by the following two lemmas below.}
\begin{lemma}
    \label{lemma:large_n_case}
    If $c \geq 150$, $m \geq 2c$, $n \geq c$, and $n \geq \frac{40}{c} \log(12 c) \cdot m$ then $\prob{\cE_1} \geq 1/2 \geq \prob{\cE_2}$.
\end{lemma}
\begin{lemma}
    \label{lemma:small_n_case}
    If $c \geq 20000$, $m \geq c$,
    and $n \leq \frac{40}{c} \log(12c) \cdot m$.
    Then $\prob{\cE_1} \geq \prob{\cE_2}$.
\end{lemma}
\begin{proof}[Proof of Lemma~\ref{lemma:probE1_gt_probE2}]
    Immediate from Lemma~\ref{lemma:large_n_case} and Lemma~\ref{lemma:small_n_case}.
\end{proof}

\subsubsection{Proof of Lemma~\ref{lemma:large_n_case}}
The main ingredient to prove Lemma~\ref{lemma:large_n_case} is the following lemma which shows that if $n = \Theta(m)$ then the probability that event $\cE_1$ does \emph{not} happen decays exponentially quickly in $c$.
\begin{lemma}
    \label{lemma:E1_prob_lb}
    If $m \geq n \geq c$ then
    $\prob{\cE_1} \geq 1 - 6c \cdot \exp\left( -\frac{cn}{10(m+n+2c)} \right)$.
\end{lemma}

Given this lemma, the proof of Lemma~\ref{lemma:large_n_case} is straightforward.
\begin{proof}[Proof of Lemma~\ref{lemma:large_n_case}]
    Since $\prob{\cE_2} \leq \prob{\neg \cE_1} = 1 - \prob{\cE_1}$, we only need to to prove that $\prob{\neg \cE_1} \leq 1/2$.
    By Lemma~\ref{lemma:E1_prob_lb}, we need to check that
    \begin{equation}
        \label{eqn:big_n_sufficient}
        6c \cdot \exp\left( - \frac{cn}{10(m+n+2c)} \right) \leq \frac{1}{2}. 
    \end{equation}
    Rearranging, this is equivalent to $n\left(1 - \frac{10}{c} \log(12c)\right) \geq \frac{10}{c} \log(12c)(m+2c)$.
    Indeed,
    \[
        n\left(1 - \frac{10}{c} \log(12c)\right) \geq \frac{n}{2} \geq \frac{20}{c} \log(12c) m \geq \frac{10}{c} \log(12c)(m+2c)
    \]
    where the first inequality uses Claim~\ref{claim:c_ineq}, the second inequality uses the assumption on $n$, and the last inequality uses that $2m = m + m \geq m + 2c$.
\end{proof}

We now prove Lemma~\ref{lemma:E1_prob_lb}.
\begin{proof}[Proof of Lemma~\ref{lemma:E1_prob_lb}]
Recall that $p=\lceil n/10\rceil$. First, we have
\[
    \Pr_{\pi}[|I_1 \cap \Bnew^{\pi}| = 0]
    = \frac{\binom{m+n+c}{p}}{\binom{m+n+2c}{p}}.
\]
To see this, one can think of the following alternative way of choosing the assignment: The agents (including buyers and sellers) are named as $1,2,\ldots, m+n+2c$. Agents $1,2,\ldots, m$ are old buyers; Agents $m+1,m+2,\ldots, m+n$ are old sellers; Agents $m+n+1,\ldots m+n+c$ are new buyers; Agents $m+n+c+1,\ldots m+n+2c$ are new sellers. We choose $p$ agents and assign them the ``In $I_1$'' label. 
%The numerator is the number of ways to assign $p$ indices to the $m+n+c$ quantiles which are \emph{not} a new buyer.
The numerator is the number of ways to choose $p$ agents from the set of ``not new buyers'' (and assign them the label). The denominator is the number of ways to choose $p$ agents from all $m+n+2c$ agents (and assign them the label).
%to any of the $m+n+2c$ buyers and sellers.
Similarly, we have
\[
    \Pr_{\pi}[|I_1 \cap \Bnew^{\pi}| = 1]
    = \frac{\binom{c}{1}\cdot \binom{m+n+c}{p - 1}}{\binom{m+n+2c}{p}}.
\]
The same argument can be applied for $|J_1 \cap \Snew^{\pi}|$.
In other words,
\[
    \Pr_{\pi}[|I_1 \cap \Bnew^{\pi}| \leq 1] = \Pr_{\pi}[|J_1 \cap \Snew^{\pi}| \leq 1]
    = \frac{\binom{m+n+c}{p} + c \cdot \binom{m+n+c}{p - 1}}{\binom{m+n+2c}{p}}.
\]
Next, we have
\[
    \Pr_{\pi}[|I_2 \cap \Bold^{\pi}| = 0]
    = \frac{\binom{n+2c}{p}}{\binom{m+n+2c}{p}}
\]
and
\[
    \Pr_{\pi}[|J_2 \cap \Sold^{\pi}| = 0]
    = \frac{\binom{m+2c}{p}}{\binom{m+n+2c}{p}}.
\]
Here we notice that the numerators in the two equations above are the number of ways to choose $p$ agents from the set excluding the $m$ old buyers (or $n$ old sellers), and assign them the ``In $I_2$ (or $J_2$)'' label. By union bound, we have
\begin{align*}
    \prob{\neg \cE_1}
    & \leq \frac{2\binom{m+n+c}{p} + 2c\cdot\binom{m+n+c}{p - 1} + \binom{n+2c}{p} + \binom{m+2c}{p} }{\binom{m+n+2c}{p}} \\
    & \leq \frac{6c \binom{m+n+c}{\lceil n / 10 \rceil}}{\binom{m+n+2c}{\lceil n / 10 \rceil}} \\
    & = 6c
        \cdot \frac{(m+n+c)!}{\lceil n / 10 \rceil! \cdot (m + n + c - \lceil n /10 \rceil)!}
        \cdot \frac{\lceil n / 10 \rceil! \cdot (m+n+2c - \lceil n / 10 \rceil)!}{(m+n+2c)!} \\
    & = 6c \prod_{i=1}^c \frac{m+n+c-\lceil n / 10 \rceil + i}{m+n+c+i} \\
    & = 6c \prod_{i=1}^c \left(1 - \frac{\lceil n / 10 \rceil}{m+n+c+i} \right) \\
    & \leq 6c\cdot \exp\left(-\sum_{i=1}^c \frac{n/10}{m+n+c+i}\right)
    \leq 6c \cdot \exp\left( - \frac{cn}{10(m+n+2c)} \right)
\end{align*}
In the second inequality, we used that (i) $p=\lceil n / 10 \rceil$, (ii) $\binom{m+n+c}{\lceil n / 10 \rceil - 1} \leq \binom{m+n+c}{\lceil n / 10 \rceil}$ since $\lceil n / 10 \rceil<\frac{1}{2}(m+n+c)$, (iii) $n + 2c\leq m+2c \leq m+n+c$ since $m \geq n \geq c$. The second-last inequality follows from $1-x\leq \exp(-x)$ for all $x>0$.
\end{proof}

%\begin{claim}
%    \label{claim:c_ineq}
%    If $c \geq 150$ then $\frac{10}{c} \log(12c) \leq 1/2$.
%\end{claim}
%\begin{proof}
%    Let $f(c) = \frac{10}{c} \log(12c)$.
%    One can check that $f(150) < 1/2$.
%    Moreover, $f'(c) = -\frac{10}{c^2} (\log(12c) - 1) > 0$ since $\log(12c) < \log(e) = 1$.
%\end{proof}

\subsubsection{Proof of Lemma~\ref{lemma:small_n_case}}
To prove Lemma~\ref{lemma:small_n_case}, we require two lemmas.
The first lemma is another lower bound on $\prob{\cE_1}$ which is tighter than Lemma~\ref{lemma:E1_prob_lb} when $n \ll m$ (the latter uses a union bound which becomes too loose in this setting).
The second lemma upper bounds $\prob{\cE_2}$ by upper bounding the probability that all new sellers are in the top $2n+2c$ quantiles.
In contrast, the previous subsection upper bounded $\prob{\cE_2}$ by simply saying that $\prob{\cE_2} \leq 1 - \prob{\cE_1}$.
However, since $\prob{\cE_1}$ is also small for small $n$. This latter upper bound is far too weak.

\begin{lemma}
    \label{lemma:small_n_E1}
    Let $\alpha > 0$ (possibly depending on $c$). Suppose that (i) $c \geq 2$, (ii) $m \geq n + 2c \geq n + c + 2$ and (iii) $n \leq \frac{10 \alpha m}{c} - 1$.
    Then
    \[
        \prob{\cE_1} \geq
        \frac{1}{40} \left( \frac{c}{120} \right)^4 \cdot \left(1 - \frac{10\alpha}{c} \right)^{2c} \cdot \frac{n^6}{m^6}.
    \]
\end{lemma}
The proof of Lemma~\ref{lemma:small_n_E1} is somewhat lengthy so we relegate it to Subsection~\ref{subsubsec:small_n_E1}.

\begin{lemma}
    \label{lemma:small_n_E2}
    %Suppose that $n \leq m / 6$.
    %Then $\prob{\cE_2} \leq (8n / m)^c$.
    {Suppose that $n \leq m / 4$.
    Then $\prob{\cE_2} \leq (4n / m)^c$.}
\end{lemma}
\begin{proof}
We have
\begin{align*}
    \Pr_{\pi}[\Snew^{\pi} \subseteq [2n + 2c]]
    & =
    \frac{\binom{m+n+c}{2n+c}}{\binom{m+n+2c}{2n+2c}} \\
    & = \frac{(m+n+c)!}{(2n+c)!(m-n)!} \cdot \frac{(2n+2c)!(m-n)!}{(m+n+2c)!} \\
    & = \prod_{i=1}^c \frac{2n+c+i}%{m-n-c+i} \\
    {m+n+c+i}
    %& \leq \frac{(4n)^c}{(m/2)^c} \\
    \leq \left( \frac{4n}{m} \right)^c
\end{align*}
In the first equality, the denominator is the number of ways to choose $2n+2c$ agents from all $m+n+2c$ agents (and assign them the ``In $[2n+2c]$'' label). When all $c$ new sellers are assigned the ``In $[2n+2c]$'' label, the numerator is the number of ways to choose another $2n+c$ agents from the rest $m+n+c$ agents (and assign them the label).
The inequality follows from $2n+c+i \leq 2n + 2c \leq 4n$ %and $m-n-c+i \geq m-n-2c \geq m-3n \geq m/2$ provided $n \leq m / 6$.
and $m+n+c+i\geq m$.
\end{proof}

\begin{proof}[Proof of Lemma~\ref{lemma:small_n_case}]
    Since $m \geq c \geq \frac{c}{40 \log(12c)}$ and $n \leq \frac{40}{c} \log(12c) \cdot m \leq m / 4$,
    it is straightforward to verify that $n \leq \frac{10 \alpha m}{c} - 1$ with $\alpha = 8 \log(12c)$.
    Thus, we have
    \begin{align*}
        \frac{\prob{\cE_2}}{\prob{\cE_1}}
        & \leq \frac{40 \left( \frac{120}{c} \right)^4}{\left( 1 - 10\alpha / c\right)^{2c}} \cdot 4^c \cdot \left( \frac{n}{m} \right)^{c-6} \\
        & \leq \frac{40 \left( \frac{120}{c} \right)^4}{\left( 1 - 10\alpha / c\right)^{2c}} \cdot 4^c \cdot \left( \frac{10\alpha}{c} \right)^{c-6} \\
        & \leq \frac{40 \left( \frac{120}{c} \right)^4}{\left( 1 - 10\alpha / c\right)^{2c}} \cdot 4^6 \cdot \left( \frac{40\alpha}{c} \right)^{c-6} \\
        & \leq 40 \left( \frac{120}{c} \right)^4 \cdot 16^6 \cdot \left( \frac{160\alpha}{c} \right)^{c-6} \\
        & \leq 9000 \cdot \left( \frac{160 \alpha}{c} \right)^{c-6} \\
        & \leq 9000 \cdot 0.8^{c-6} \\
        & < 1.
    \end{align*}
    In the fourth inequality, we used that $1 - 10\alpha / c = 1 - 80 \log(12c) / c \geq 1/2$ for $c \geq 2000$.
    In the sixth inequality, we used that $1280 \log(12c) / c \leq 0.8$ for $c \geq 20000$ (see Claim~\ref{claim:c_ineq}).
\end{proof}

\subsubsection{Proof of Lemma~\ref{lemma:small_n_E1}}
\label{subsubsec:small_n_E1}
\begin{claim}
    \label{claim:independence_lb}
    For the event $\cE_1$, we have
    \[
        \prob{\cE_1} \geq \prob{|I_1 \cap \Bnew \geq 2} \cdot \prob{|I_2 \cap \Bold| \geq 1} \cdot \prob{|J_2 \cap \Sold| \geq 1} \cdot \prob{|J_1 \cap \Snew| \geq 2}.
    \]
\end{claim}
\begin{proof}
    Recall that
    \[
        \cE_1 = \{ |I_1 \cap \Bnew| \geq 2, |I_2 \cap \Bold| \geq 1, |J_2 \cap \Sold| \geq 1, |J_1 \cap \Snew| \geq 2 \}.
    \]
    We first check that
    \[
        \prob{|I_1 \cap \Bnew| \geq 2 \mid |I_2 \cap \Bold| \geq 1, |J_2 \cap \Sold| \geq 1, |J_1 \cap \Snew| \geq 2}
        \geq \prob{|I_1 \cap \Bnew| \geq 2}.
    \]
    Let $\cE_{1,1} = \{ |I_2 \cap \Bold| \geq 1, |J_2 \cap \Sold| \geq 1, |J_1 \cap \Snew \geq 2 \}$
    and $\cF = \{ \exists X \subseteq [N] \setminus I_1, |X| = 4, X \cap \Bnew = \emptyset \}$.
    Note that $\cE_{1,1} \subseteq \cF$ and that the random variable $|I_1 \cap \Bnew|$ is independent of $\cE_{1,1}$ given $\cF$.
    Thus,
    \begin{align*}
        \probg{|I_1 \cap \Bnew| \geq 2}{\cE_{1,1}} 
        & = \probg{|I_1 \cap \Bnew| \geq 2}{\cE_{1,1} \wedge \cF} \\
        & = \probg{|I_1 \cap \Bnew| \geq 2}{\cF} \\
        & \geq \prob{|I_1 \cap \Bnew| \geq 2}.
    \end{align*}
    where the last inequality is by Claim~\ref{claim:prob_condition_lb}.
    % Note that the LHS is a conditional probability of $|I_1 \cap \Bnew| \geq 2$ conditioned on $|I_2 \cap \Bold| \geq 1, |J_2 \cap \Sold| \geq 1, |J_1 \cap \Snew| \geq 2$. To compute the probability on the left, we can first remove $4$ quantiles from $I_2 \cup J_2 \cup J_1$.
    % Then the LHS is the probability that at least $2$ out of $c$ quantiles chosen uniformly at random from the remaining $m+n+2c-4$ quantiles are in $I_1$.
    % On the other hand, the RHS probability is the probability that at least $2$ out of $c$ quantiles chosen uniformly at random from all $m+n+2c$ quantiles are in $I_1$.
    % In particular, the former probability is larger than the latter probability.
    % Thus, we conclude that
    % \[
    %     \prob{\cE_1} \geq \prob{|I_1 \cap \Bnew| \geq 2} \cdot \prob{|I_2 \cap \Bold| \geq 1, |J_2 \cap \Sold| \geq 1, |J_1 \cap \Bold| \geq 2}.
    % \]
    Continuing this argument gives the claim.
\end{proof}

\begin{claim}
    \label{claim:small_n_I1_J1_lb}
    Let $\alpha > 0$. Suppose that (i) $c \geq 2$, (ii) $m \geq n + c + 2$ and (iii) $n \leq \frac{10 \alpha m}{c} - 1$.
    Then
    \[
        \prob{|I_1 \cap \Bnew| \geq 2} = \prob{|J_1 \cap \Snew| \geq 2} \geq \frac{c^2}{12800} \cdot \frac{n^2}{m^2} \cdot \left( 1 - \frac{10\alpha}{c} \right)^c.
    \]
\end{claim}
\begin{proof}
    We compute a lower bound on $\prob{|I_1 \cap \Bnew| = 2}$.
    We have that
    \begin{align*}
        \prob{|I_1 \cap \Bnew| = 2}
        & = \frac{\binom{c}{2} \cdot \binom{m+n+c}{\lceil n / 10 \rceil - 2}}{\binom{m+n+2c}{\lceil n / 10 \rceil}} \\
        & = \binom{c}{2} \cdot \frac{(m+n+c)!}{(m+n+2c)!} \cdot \frac{\lceil n / 10 \rceil! (m+n+2c - \lceil n / 10 \rceil)!}{(\lceil n / 10 \rceil - 2)!(m+n+c - \lceil n / 10 \rceil + 2)!} \\
        & \geq \frac{(c-1)^2}{2} \cdot \frac{\left( \frac{n}{10} - 1 \right)^2}{(m+n+c+2)^2} \cdot \prod_{i=3}^c \frac{m+n+c-\lceil n / 10 \rceil + i}{m+n+c+i} \\
        & \geq \frac{(c-1)^2}{2} \cdot \frac{\left( \frac{n}{10} - 1 \right)^2}{(m+n+c+2)^2} \cdot \left( 1 - \frac{\lceil n / 10 \rceil}{m+n+c} \right)^{c-2} \\
        & \geq \frac{(c-1)^2}{2} \cdot \frac{\left( \frac{n}{10} - 1 \right)^2}{(m+n+c+2)^2} \cdot \left( 1 - \frac{10\alpha m / c}{m+n+c} \right)^{c-2} \\
        & \geq \frac{(c-1)^2}{2} \cdot \frac{\left( \frac{n}{10} - 1 \right)^2}{(m+n+c+2)^2} \cdot \left( 1 - \frac{10\alpha}{c} \right)^c \\
        & \geq \frac{c^2}{8} \cdot \frac{(n/20)^2}{(2m)^2} \cdot \left( 1 - \frac{10\alpha}{c} \right)^c \\
        & = \frac{c^2}{12800} \cdot \frac{n^2}{m^2} \cdot \left( 1 - \frac{10\alpha}{c} \right)^c.
    \end{align*}
    In the third inequality, we used that $\lceil n / 10 \rceil < n/10 + 1 \leq 10\alpha m / c$ since $n \leq 10\alpha m / c - 1$.
    In the fourth inequality, we used the trivial inequality $m + n + c \geq m$.
    Finally, in the fifth inequality, we used that $c-1 \geq c/2$, $n/10 - 1 \geq n / 20$ and $m+n+c+2 \leq 2m$.
\end{proof}

\begin{claim}
    \label{claim:small_n_I2_J2_lb}
    If $m \geq n + 2c$ then $\prob{|I_2 \cap \Bold| \geq 1} \geq \prob{|J_2 \cap \Sold| \geq 1} \geq \frac{n}{20m}$.
\end{claim}
\begin{proof}
    The first inequality is because $m \geq n$ so we prove only the second inequality.
    We have that
    \begin{align*}
        \prob{|J_2 \cap \Sold| = 0}
        & = \frac{\binom{m+2c}{\lceil n / 10 \rceil}}{\binom{m+n+2c}{\lceil n / 10 \rceil}} \\
        & = \frac{(m+2c)!}{(m+n+2c)!} \cdot \frac{(m+n+2c-\lceil n / 10 \rceil)!}{m+2c-\lceil n / 10 \rceil)!} \\
        & = \prod_{i=1}^n \frac{m+2c+i-\lceil n / 10 \rceil}{m+2c+i} \\
        & = \prod_{i=1}^n \left( 1 - \frac{\lceil n / 10 \rceil}{m+2c+i} \right) \\
        & \leq \prod_{i=1}^n \left( 1 - \frac{n}{10(m+n+2c)} \right) \\
        & \leq 1 - \frac{n}{10(m+n+2c)}
    \end{align*}
    In the first inequality, we used that $\lceil n / 10 \rceil \geq n / 10$ and $m=2c+i \leq m+n+2c$ for all $i \in [n]$.
    Thus, $\prob{|I_2 \cap \Bold| \geq 1} \geq \frac{n}{10(m+n+2c)} \geq \frac{n}{20m}$, where the final inequality is because $n+2c \leq m$.
\end{proof}

\begin{proof}[Proof of Lemma~\ref{lemma:small_n_E1}]
    Follows directly by combining Claim~\ref{claim:independence_lb}, Claim~\ref{claim:small_n_I1_J1_lb}, and Claim~\ref{claim:small_n_I2_J2_lb}.
\end{proof}

\subsection{Comparison Between Trade Reduction and STR}
\label{app:compare_tr_str}
As mentioned in \Cref{subsec:fsd_event_construction}, our argument crucially makes use the fact that we use STR instead of the classic Trade Reduction mechanisms \cite{McAfee92}.
In particular, a key observation we use is that, if (i) the optimal allocation in the augmented market, $\OPT'$ is not the same as the optimal allocation in the original market $\OPT$ and (ii) the size of the optimal matching remains the same then $\STR$ and $\OPT'$ have the same GFT.
This would not be true using McAfee's Trade Reduction mechanism \cite{McAfee92}.
As an instructive example, consider the following scenario (we will assume the both sides have the same distribution so that values and quantiles are consistent).
We have one original buyer with value $1$, one original seller with value $0.9$, one new buyer with value $0$, and one new seller with value $0.8$.
In this example, the original first-best matching has size $1$ and a GFT of $0.1$.
Once we add in the new buyers and sellers, the first-best matching remains at size $1$ but the GFT is now $0.2$.
In STR, we use the second-lowest value seller to price the buyers and sellers, if possible.
Here, this means a price of $0.9$ is offered to buyer with value $1$ and the seller with value $0.8$ resulting in a trade.
On the other hand, the Trade Reduction mechanism offers a price equal to the \emph{average} of the next unmatched buyer and seller.
This means a price of $(0 + 0.9) / 2 = 0.45$ is offered to the buyer and seller.
Since the seller will not accept this price, the match is reduced resulting in zero trades.

It is not too difficult to extend the above example that show that $\cE_2$ is not a necessary condition for OPT to outperform Trade Reduction.
Concretely, suppose we have $n$ original buyers with value $2$, $n-1$ original sellers with value $1$, one original seller with value $1+\eps$, and $2c$ original buyers with value $0.9$ for a total of $2n + 2c$.
Here, the first-best allocation trades the $n$ value-$2$ buyers with all sellers, resulting in a GFT of $n + \eps$.
We then add $c$ new buyers with value $0$, $c-1$ new sellers with value $100$, and one new seller with value $0.8$.
In particular, $\cE_2$ does not happen since one of the new seller is outside the top $2n+2c$ values (and thus quantiles) in the augmented market.
In the augmented first-best matching, the seller with value $1+\eps$ would be removed from the matching and the new seller with value $0.8$ would be added to the matching.
STR would then offer a price of $1+\eps$ which is accepted by the buyers with value $2$ and the sellers with value at most $1$.
On the other hand, Trade Reduction would offer a price of $(0.9 + 1+\eps)/2 < 1$ if $\eps < 0.1$.
This price is not accepted by those sellers with value 1. Thus in Trade Reduction, $n-1$ buyers with value $2$ will trade with one (new) seller with value $0.8$ and $n-2$ (old) sellers with value $1$. The GFT of Trade Reduction would be $2(n-1) - (n-2) - 0.8 = n-0.8 < n$, which is worse than the original first-best GFT.

\subsection{Example Where New Agents Can Trade but STR Loses a Trade}
\label{app:str_loses_trade_example}
In this short section, we give an example where the new agents can trade but STR is still worse than OPT in the augmented market.
The example is also depicted in Figure~\ref{figure:str_lose_trade}.
Let $\eps > 0$.
There are $3$ original buyers with values $\oldb_1 = 3, \oldb_2 = 2+\eps, \oldb_3 = 2$ and $3$ original sellers all with value $\olds_1 = \olds_2 = \olds_3 = 1$.
The original optimal GFT is $4 + \eps$.
Now, suppose we add a new buyer with value $\newb = 2+3\eps$ and a new seller with value $\news = 2 + 2\eps$.
Note that $\newb$ and $\news$ are eligible to trade with each other.
The new optimal matching matches $\oldb_1, \newb, \oldb_2$ with $\olds_1, \olds_2, \olds_3$.
STR checks if $\news$ is able to price the buyers and the sellers; in this case it has higher value than $\oldb_2$.
Thus, STR removes $\oldb_2$ and $\olds_3$ (say) and matches only $\oldb_1, \newb$ with $\olds_1, \olds_2$ for a GFT of $3 + 3\eps$.
This is strictly worse than the original GFT if $\eps < 1/2$.
Note that a similar example could be possible even if there are multiple trades among the new agents.
\begin{figure}[ht]
\begin{center}
\resizebox{200pt}{!}{%
\begin{tikzpicture}[
oldbuyer/.style={circle, draw=blue!60, fill=blue!5, very thick, minimum size=9mm},
newbuyer/.style={rectangle, draw=red!60, fill=red!5, very thick, minimum size=8mm},
oldseller/.style={circle, draw=blue!60, fill=blue!5, very thick, minimum size=9mm},
newseller/.style={rectangle, draw=red!60, fill=red!5, very thick, minimum size=8mm},
]

%% Original example
%Nodes
\node at (0, 4) {\textbf{Buyers}};
\node at (3, 4) {\textbf{Sellers}};
\node[oldbuyer] (oldbuyer11) at (0, 3) {$3$};
\node[oldbuyer] (oldbuyer12) at (0, 1) {$2.1$};
\node[oldbuyer] (oldbuyer13) at (0, 0) {$2$};

\node[oldseller] (oldseller11) at (3, -1) {$1$};
\node[oldseller] (oldseller12) at (3, -2) {$1$};
\node[oldseller] (oldseller13) at (3, -3) {$1$};
\node at (1.5, -4) {$\text{OPT} = 4.1$};

%Lines
\draw[black, very thick] (oldbuyer11.south east) -- (oldseller11.north west);
\draw[black, very thick] (oldbuyer12.south east) -- (oldseller12.north west);
\draw[black, very thick] (oldbuyer13.south east) -- (oldseller13.north west);

%% New example
%Nodes
\node at (6+0, 4) {\textbf{Buyers}};
\node at (6+3, 4) {\textbf{Sellers}};
\node[oldbuyer] (oldbuyer21) at (6+0, 3) {$3$};
\node[oldbuyer] (oldbuyer22) at (6+0, 1) {$2.1$};
\node[oldbuyer] (oldbuyer23) at (6+0, 0) {$2$};
\node[newbuyer] (newbuyer)   at (6+0, 2) {$2.3$};

\node[oldseller] (oldseller21) at (6+3, -1) {$1$};
\node[oldseller] (oldseller22) at (6+3, -2) {$1$};
\node[oldseller] (oldseller23) at (6+3, -3) {$1$};
\node[newseller] (newseller)   at (6+3, 1.5) {$2.2$};
\node at (6 + 1.5, -4) {$\text{OPT} = 4.4$, $\text{STR} = 3.3$};

%Lines
\draw[black, very thick] (oldbuyer21.south east) -- (oldseller21.north west);
\draw[black, very thick] (newbuyer.south east) -- (oldseller22.north west);
\draw[black, very thick, dotted] (oldbuyer22.south east) -- (oldseller23.north west);

\draw[black, very thick] (4.5, 4.5) -- (4.5, -4.5);

\end{tikzpicture}
}
\end{center}
\caption{This figure illustrates an example where the the new agents can trade yet incorporating them with the original agents causes STR to lose a trade. In the figure, \textcolor{blue}{original agents} are depicted with \textcolor{blue}{blue circles} and \textcolor{red}{new agents} are depicted with \textcolor{red}{red squares}. The figure on the left depicts the original market and the figure on the right depicts the augmented market.}
\label{figure:str_lose_trade}
\end{figure}
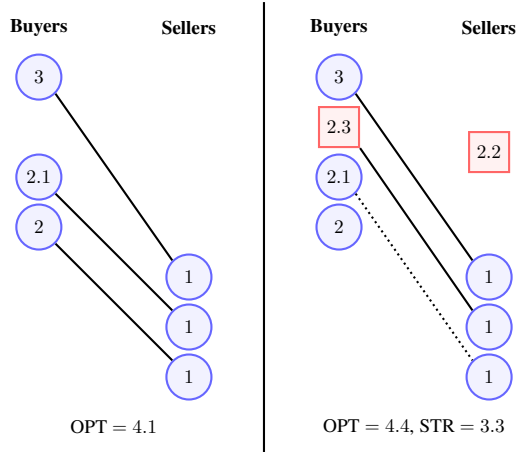
\section{Missing Proofs From \Cref{sec:fsd_noniid_approx}}
\subsection{Proof of \Cref{lemma:combinatorial_fsd_approx}}
\label{subsec:combinatorial_fsd_approx}
\begin{proof}[Proof of \Cref{lemma:combinatorial_fsd_approx}]
    We consider the following probabilistic construction of $f$.
    For each $t \in [T]$ and $S \in \binom{[n]}{\leq \alpha n}$, we set $f_t(S)$ to the set which includes every element in $[n]$ with probability $\gamma$ (independently) union with $S$.
    We now split the proof into three claims which verify each requirement of the lemma.
    \begin{claim}
        \label{claim:combinatorial1}
        With probability at least $1 - (e/\alpha)^{\alpha c} e^{-c \gamma / 12}$, we have $|f_t(S)| \geq \gamma c / 2$ for all $t \in [T]$ and $S \in \binom{[c]}{\leq \alpha c}$.
    \end{claim}
    \begin{proof}
        By definition of $f_t(S)$, we have that $\expect{|f_t(S)|} \geq c \gamma$.
        Thus, a standard Chernoff bound (see \Cref{lemma:chernoff}) gives that $\prob{|f_t(S)| \leq c\gamma / 2} \leq e^{-c\gamma / 12}$. By \Cref{fact:binom_bounds}, we have $|\binom{[c]}{\leq \alpha c}| \leq (e/\alpha)^{\alpha c}$.
        The claim now follows by a union bound.
    \end{proof}

    \begin{claim}
        \label{claim:combinatorial2}
        With probability at least $1-T^2 (e/\alpha)^{2\alpha c} (\gamma^2 + (1-\gamma)^2)^{c(1-2\alpha)}$, we have $f_{t_1}(S_1) \neq f_{t_2}(S_2)$ for all $t_1, t_2 \in [T]$ and $S_1, S_2 \in \binom{[c]}{\leq \alpha c}$ provided $(t_1, S_1) \neq (t_2, S_2)$.
    \end{claim}
    \begin{proof}
        We have that
        \[
            \prob{f_i(S_1) = f_j(S_2)} \leq  (\gamma^2 + (1-\gamma)^2)^{c(1-2\alpha)}.
        \]
        Taking a union bound over all pairs of $(i,S_1)$ and $(j, S_2)$ gives that
        \begin{align*}
            \prob{\exists (i, S_1), (j, S_2) \text{ such that }f_i(S_1) = f_j(S_2)}
            \leq T^2 (e/\alpha)^{2\alpha c} (\gamma^2 + (1-\gamma)^2)^{c(1-2\alpha)},
        \end{align*}
        where in the last inequality we used that $|\binom{[c]}{\leq \alpha c}| \leq (e/\alpha)^{\alpha c}$ (see \Cref{fact:binom_bounds}).
    \end{proof}
    \begin{claim}
        \label{claim:combinatorial3}
        Suppose that $T \geq 12c \left( \frac{1-\gamma}{\gamma^2 + (1-\gamma)^2} \right)^c$.
        Then with probability at least $1 - (e/\alpha)^{\alpha c} e^{-c}$, we have
        we have
        \[
            c \cdot \gamma^{|S|} (1-\gamma)^{c - |S|} \leq \sum_{t \in [T]} \gamma^{|f_t(S)|} (1-\gamma)^{c-|f_t(S)|}
        \]
        for every $S \in \binom{[c]}{\leq \alpha c}$ .
    \end{claim}
    \begin{proof}
        For any fixed $t$, we have that
        \begin{align*}
            \expect{\gamma^{|f_t(S)|} \cdot (1-\gamma)^{c - |f_t(S)|}} 
            & = \gamma^{|S|} \expect{\gamma^{|f_t(S)|-|S|} \cdot (1-\gamma)^{(c-|S|) - (|f_t(S)|-|S|)}} \\
            & = \gamma^{|S|} \sum_{k=0}^{c-|S|} \gamma^k \cdot (1-\gamma)^{(c-|S|)-k} \prob{|f_t(S)| - |S| = k} \\
            & = \gamma^{|S|} \sum_{k=0}^{c-|S|} \binom{c-|S|}{k} \gamma^{2k} (1-\gamma)^{2(c-|S| - k)} \\
            & = \gamma^{|S|} \cdot (\gamma^2 + (1-\gamma)^2)^{c - |S|}.
        \end{align*}
        By a Chernoff bound (see \Cref{lemma:chernoff}), since $\gamma^{|f_t(S)|}(1-\gamma)^{c-|f_t(S)|} \in [0, (1-\gamma)^{c-|S|}]$, we have
        \[
            \sum_{t \in [T]} \gamma^{|f_t(S)|} \cdot (1-\gamma)^{c - |f_t(S)|}
            \leq \frac{1}{2} T \gamma^{|S|} \cdot (\gamma^2 + (1-\gamma)^2)^{c - |S|}
        \]
        with probability at most $\exp\left( -\frac{T(\gamma^2 + (1-\gamma)^2)^{c-|S|}}{12(1-\gamma)^{c-|S|}} \right)$.
        % \textcolor{blue}{(Yang: Do we need T to be at least $12n \left( \frac{1-\gamma}{\gamma^2 + (1-\gamma)^2} \right)^n$?)}
        If $T \geq 12c \left( \frac{1-\gamma}{\gamma^2 + (1-\gamma)^2} \right)^c \geq 12c \left( \frac{1-\gamma}{\gamma^2 + (1-\gamma)^2} \right)^{c-|S|}$ (the second inequality uses $(1-\gamma) / (\gamma^2 + (1-\gamma)^2) \geq 1$ which is true when $\gamma \leq 1/2$) then we have
        \begin{equation}
            \label{eqn:combinatorial1}
            \sum_{t \in [T]} \gamma^{|f_t(S)|} \cdot (1-\gamma)^{c - |f_t(S)|} < c \cdot \gamma^{|S|} (1-\gamma)^{c-|S|}
        \end{equation}
        with probability at most $e^{-c}$.
        Taking a union bound over all $S$ shows that \Cref{eqn:combinatorial1} happens with probability at most $(e/\alpha)^{\alpha c} \exp(-c / 6)$.
    \end{proof}
    We take
    \[
        \alpha = \min\left\{\frac{\gamma / 24}{1 + 2\log(24 / \gamma)}, \frac{\log(1 + \gamma^2 / (1-\gamma)^2)/8}{1 + 2\log(8 / \log(1 + \gamma^2 / (1-\gamma)^2))}, \frac{\log(1 + \gamma^2 / (1-\gamma)^2)}{8 \log(1/(\gamma^2 + (1-\gamma)^2))} \right\} = \Theta(\gamma^2).
    \]
    % With this, note that $(e/\alpha)^\alpha \leq \gamma / 24$ (by Claim~\ref{claim:alpha_log_e_alpha}), $(e/\alpha)^{2\alpha} \leq (1 + \gamma^2 / (1-\gamma)^2)^{1/4} = ((1-\gamma)^2 + \gamma^2) / (1-\gamma)^2)^{1/4}$, and $1/(\gamma^2 + (1-\gamma)^2)^{2\alpha} \leq ((1-\gamma)^2 + \gamma^2) / (1-\gamma)^2)^{1/4}$.
    We also take
    \[
        T = 12c \left( \frac{1-\gamma}{(\gamma^2 + (1-\gamma)^2)} \right)^c.
    \]
    Let $\cE_1, \cE_2, \cE_3$ correspond to the three conditions in the lemma.
    We show that if $c$ is sufficiently large (in particular, if $c \geq \Theta(1/\gamma^2)$) then we have $\prob{\cE_1, \cE_2, \cE_3} > 0$.

    First, by \Cref{claim:combinatorial1}, we have
    \[
        \prob{\cE_1}
        \geq 1 - (e/\alpha)^{\alpha c} e^{-c\gamma / 12}
        \geq 1 - e^{-c\gamma / 24},
    \]
    where in the second inequality we used Claim~\ref{claim:alpha_log_e_alpha} and the fact that $\alpha \leq \frac{\gamma / 24}{1 + 2 \log(24/\gamma)}$ to get that $(e/\alpha)^{\alpha} \leq e^{\gamma / 24}$ (in particular, we applied Claim~\ref{claim:alpha_log_e_alpha} with $x = \gamma / 24$).

    Next, by \Cref{claim:combinatorial2} and our choice of $T$, we have
    \begin{align*}
        \prob{\cE_2}
        & \geq
        1 - 144c^2 \left( \frac{1-\gamma}{\gamma^2 + (1-\gamma)^2} \right)^{2c} \cdot (e/\alpha)^{2\alpha c} \cdot (\gamma^2 + (1-\gamma)^2)^{c(1-2\alpha)} \\
        & \geq
        1 - 144c^2 \left( \frac{1-\gamma}{\gamma^2 + (1-\gamma)^2} \right)^{2c} \cdot \left(1 + \frac{\gamma^2}{(1-\gamma)^2}\right)^{c/4} \cdot (\gamma^2 + (1-\gamma)^2)^{c(1-2\alpha)} \\
        & \geq
        1 - 144c^2 \left( \frac{1-\gamma}{\gamma^2 + (1-\gamma)^2} \right)^{2c} \cdot \left( \frac{(1-\gamma)^2 + \gamma^2}{(1-\gamma)^2} \right)^{c / 2} \cdot (\gamma^2 + (1-\gamma)^2)^{c} \\
        & =
        1 - 144c^2 \left( \frac{(1-\gamma)^2}{\gamma^2 + (1-\gamma)^2} \right)^{c/2},
    \end{align*}
    In the second inequality, we used the second term in the definition of $\alpha$ and Claim~\ref{claim:alpha_log_e_alpha} with $x = \frac{1}{8} \log\left( 1 + \frac{\gamma^2}{(1-\gamma)^2} \right)$ to bound $(e/\alpha)^{\alpha} \leq \left( 1 + \frac{\gamma^2}{(1-\gamma)^2} \right)^{1/8}$.
    In the third inequality, we used the third term in the definition of $\alpha$ to bound $(\gamma^2 + (1-\gamma)^2)^{-2\alpha} \leq \left(1 + \frac{\gamma^2}{(1-\gamma)^2} \right)^{1/4}$. We also simplified and wrote $1 + \frac{\gamma^2}{(1-\gamma)^2} = \frac{(1-\gamma)^2 + \gamma^2}{(1-\gamma)^2}$.

    Finally, we use \Cref{claim:combinatorial3} to get that
    \[
        \prob{\cE_3}
        \geq 1 - (e/\alpha)^{\alpha c} e^{-c}
        \geq 1 - e^{-c(1-\gamma/24)} \geq 1 - e^{-c\gamma/24},
    \]
    where the second inequality uses the first term in the definition of $\alpha$ and Claim~\ref{claim:alpha_log_e_alpha} with $x = \gamma / 24$ (as in the bound for $\prob{\cE_1}$) and the third inequality uses that $1-\gamma / 24 \geq \gamma / 24$ (recall $\gamma \leq 1/2$).

    % By \Cref{claim:combinatorial1}, \Cref{claim:combinatorial2}, and \Cref{claim:combinatorial3}, we have $\prob{\cE_1} \geq 1 - e^{-c\gamma / 24}$,
    % \begin{align*}
    %     \prob{\cE_2}
    %     & \geq
    %     1 - 144c^2 \left( \frac{1-\gamma}{\gamma^2 + (1-\gamma)^2} \right)^{2c} \cdot \left( \frac{(1-\gamma)^2 + \gamma^2}{(1-\gamma)^2} \right)^{c / 2} \cdot (\gamma^2 + (1-\gamma)^2)^{-c} \\
    %     & =
    %     1 - 144c^2 \left( \frac{(1-\gamma)^2}{\gamma^2 + (1-\gamma)^2} \right)^{c/2},
    % \end{align*}
    % and
    % $\prob{\cE_3} \geq 1 - e^{-c(1-\gamma / 24)} \geq 1 - e^{-c\gamma / 24}$.
    By a union bound, we have that
    \[
        \prob{\cE_1, \cE_2, \cE_3}
        \geq 1 - 2e^{-c\gamma / 24} - 144c^2 \left( \frac{(1-\gamma)^2}{\gamma^2 + (1-\gamma)^2} \right)^{c/2}.
    \]
    Define $a = \frac{2 \log(576)}{\log((\gamma^2 + (1-\gamma)^2) / (1-\gamma)^2)}$ and $b = \frac{4}{\log((\gamma^2 + (1-\gamma)^2) / (1-\gamma)^2)}$. Note that $a, b = \Theta(1/\gamma^2)$.
    We now take
    \[
        c \geq \max\left\{ \frac{24 \log 8}{\gamma},  2\max\{a, 5(b+1)(1 + \log(b+1))\} \right\} = \Theta(1/\gamma^2).
    \]
    In this case, we have $2e^{-c\gamma / 24} \leq 1/4$.
    Further, some calculations (see \Cref{claim:log_ineq1}) gives that
    \begin{align*}
        c
        & \geq a + b \log c \\
        & = \frac{2 \log 576}{\log((\gamma^2 + (1-\gamma)^2) / (1-\gamma)^2)} + \frac{4 \log c}{\log((\gamma^2 + (1-\gamma)^2) / (1-\gamma)^2)},
    \end{align*}
    which, upon rearranging, is equivalent to
    \[
        144c^2 \left( \frac{(1-\gamma)^2}{\gamma^2 + (1-\gamma)^2} \right)^{c/2} \leq \frac{1}{4}.
    \]
    We conclude that $\prob{\cE_1, \cE_2, \cE_3} > 0$.
\end{proof}
\begin{fact}[{\cite[Exercise~2.14]{boucheron2013concentration}}]
    \label{fact:binom_bounds}
    For all $c \geq 1$ and $1 \leq k \leq c$, we have $\sum_{j = 0}^k \binom{c}{k} \leq \left( \frac{ec}{k} \right)^k$.
\end{fact}

\begin{claim}
    \label{claim:alpha_log_e_alpha}
    Fix $x \in (0, 1)$.
    Suppose that $0 <\alpha \leq \frac{x}{1 + 2\log(1/x)}$.
    Then $\alpha \log(e/\alpha) \leq x$.
    Equivalently, $(e/\alpha)^{\alpha} \leq e^x$.
\end{claim}
\begin{proof}
    First we check that $\frac{x}{1 + 2 \log(1/x)} < 1$ whence $\alpha < 1$.
    Let $f(x) = \frac{x}{1+2\log(1/x)}$.
    Then $f'(x) = \frac{2\log(1/x) + 3}{(1+2\log(1/x))^2} > 0$ when $x \in (0, 1)$.
    Thus, $f(x) < f(1) = 1$ when $x \in (0, 1)$.
    Now let $g(\alpha) = \alpha \log(e/\alpha)$.
    Then $g'(\alpha) = - \log(\alpha) > 0$ when $\alpha < 1$ so it suffices to check the claim only when $\alpha = \frac{x}{1 + 2\log(1/x)}$.
    In this case, we have
    \begin{align*}
        \alpha \log(e/\alpha)
        & = \frac{x}{1 + 2\log(1/x)} \left[1 + \log\left( \frac{1 + 2\log(1/x)}{x} \right) \right] \\
        & \leq \frac{x}{1 + 2\log(1/x)} \cdot [1 + 2\log(1/x)] = x,
    \end{align*}
    where the inequality is from \Cref{claim:alpha_log_e_alpha_aux}.
\end{proof}
\begin{claim}
    \label{claim:alpha_log_e_alpha_aux}
    If $x \in (0,1)$ then $\log((1+2\log(1/x))/x) \leq 2\log(1/x)$.
\end{claim}
\begin{proof}
    By exponentiating, the inequality is equivalent to $1 + 2 \log(1/x) \leq 1/x^3$.
    Taking $u = 1/x$ the inequality is equivalent to $1 + 2\log u \leq u^3$ for $u > 1$.
    Let $f(u) = u^3 - 2 \log u - 1$.
    Then $f(1) = 0$ and $f'(u) = 3u^2 - 2 / u > 0$ when $u \geq 1$.
    We conclude that $f(u) \geq 0$ for all $u \geq 1$.
\end{proof}

\begin{claim}
    \label{claim:log_ineq1}
    Let $a, b > 0$.
    If $x \geq 2 \max\{a, 5(b+1)(1 + \log(b+1))\}$ then $x \geq a + b\log x$.
    Equivalently, $e^x / x^b \geq e^a$.
\end{claim}
\begin{proof}
    We have $x - b\log x \geq x/2 \geq a$ where the first inequality follows from \Cref{claim:log_ineq2} and the second inequality is because $x \geq 2a$.
\end{proof}
\begin{claim}
    \label{claim:log_ineq2}
    Fix $b > 0$.
    If $x \geq 10(b+1)(1 + \log(b+1))$ then $x - b \log x \geq x / 2$.
\end{claim}
\begin{proof}
    The last inequality is equivalent to $x/2 - b \log x \geq 0$.
    Let $f(x) = x/2 - b\log x$.
    Note that $f'(x) = 1/2 - b / x$ so $f$ is increasing on $(2b, \infty)$.
    Thus it suffices to prove that $f(x) \geq 0$ when $x= 10(b+1)(1 + \log(b+1))$.
    Let
    \begin{align*}
        g(b)
        & = f(2(b+1)(1 + \log(b+1))) \\
        & = 5(b+1) + 5(b+1) \log(b+1) - b \log(2) - b \log(b+1) - b\log(1 + \log(b+1)).
    \end{align*}
    Some calculations give that
    \begin{align*}
        g'(b) & = 5 + 5(\log(b+1) + 1) - \log(2) - \frac{b}{b+1} - \log(b+1) \\
        & - \frac{b}{(b+1)(1 + \log(b+1))} - \log(1+\log(b+1)).
    \end{align*}
    Differentiating again gives
    \begin{align*}
        g''(b) & = \frac{5}{b+1} + \frac{1}{(b+1)^2} - \frac{1}{b+1} + \frac{b - \log(b+1) + 1}{(b+1)^2(1+\log(b+1))^2} - \frac{1}{(b+1)(1+\log(b+1))} \\
        & = \frac{4}{b+1} + \frac{1}{(b+1)^2} - \frac{(b+2) \log(b+1)}{(b+1)^2(1 + \log(b+1))^2} \\
        & = \frac{4(b+1)(1 + 2 \log(b+1) + \log^2(b+1)) - (b+2) \log(b+1)}{(b+1)^2(1 + \log(b+1))^2} + \frac{1}{(b+1)^2} > 0.
    \end{align*}
    Note that $g(0) = 5$ and $g'(0) = 10 - \log 2 > 0$ and $g'(b) > 0$ for all $b > 0$.
    We conclude that $g(b) \geq g(0) > 0$ for all $b > 0$.
\end{proof}

\subsection{Lower Bound for Market Augmentation}
% Proof of \Cref{prop:approx_lb}}
\begin{proposition}
    \label{prop:approx_lb}
    For any $\gamma \in (0,1/2)$, there exists a distributions $F_B$ and $F_S$ such that $F_B^{-1}(1-\gamma) \geq F_S^{-1}(\gamma)$ and the following statement holds.
    If a market has $c$ buyers whose value distributions are drawn from $F_B$ and $c$ sellers whose value distributions are drawn from $F_S$ than running a Trade Reduction mechanism obtains a $(1-\Omega(1/\gamma c))$-approximation to the optimal GFT.
\end{proposition}
\begin{proof}
% [Proof of \Cref{prop:approx_lb}]
    Let $u \in [0,1]$ be any parameter for the Trade Reduction mechanism defined in Definition~\ref{defn:trade_reduction}.

    We define $F_B$ and $F_S$ as follows.
    \[
        F_B = \begin{cases}
            3 & \text{with probability $\gamma$} \\
            1 & \text{with probability $1-\gamma$}
        \end{cases}
        \quad \text{and} \quad
        F_S = \begin{cases}
            0 & \text{with probability $\gamma$} \\
            2 & \text{with probability $1-\gamma$}
        \end{cases}.
    \]
    Let $X_i$ be the number of agents with value $i$ for $i \in \{0, 1, 2, 3\}$.

    First, we check how often TR loses a trade.
    We consider two cases and show that TR with any value $u$ must lose a trade in at least one of two cases.
    For both cases, we assume that (i) $X_3 \neq X_0$, (ii) $\max\{X_3, X_0\} \leq c - 1$, and (iii) $\min\{X_0, X_3\} \geq 1$.
    Let $3 = b_1 \geq \ldots \geq b_c = 1$ be the buyers' values and $0 = s_1 \leq \ldots \leq s_c = 2$ be the sellers' values.
    
    \textbf{Case 1: $X_3 > X_0$.~}
    In this case, the optimal matching has size $r = X_3$ since $b_r = 3$, $s_r = s_{r+1} = 2$, and $b_{r+1} = 1$.
    According to the Trade Reduction Mechanism, buyer $r$ and seller $r$ are in the matching if and only if $3 \geq u + 2(1-u) \geq 2$.
    In other words, if $u \in (0, 1]$ then TR loses the $r$th trade.
    
    \textbf{Case 2: $X_3 < X_0$.~}
    In this case, the optimal matching has size $r = X_0$ since $b_r = b_{r+1} = 1$, $s_r = 0$, and $s_{r+1} = 2$.
    According to the Trade Reduction Mechanism, buyer $r$ and seller $r$ are in the matching if and only if $1 \geq u + 2(1-u) \geq 0$.
    These inequalities are only jointly satisfied when $u = 1$.

    In particular, trade reduction loses a trade in at least one of the above two cases with a GFT value $1$.

    We now compute the probability that TR does lose a trade.
    First, observe that $X_3 - X_0$ is essentially distributed as a $c$ step random walk that does not move with probability $(1-\gamma)^2 + \gamma^2$ and takes a $\pm 1$ step uniformly at random with probability $2\gamma - 2\gamma^2 = \Theta(\gamma)$.
    Thus, for $c = \Omega(1/\gamma)$ sufficiently large, a standard Chernoff bound shows that this random walk takes $\Omega(\gamma c)$ non-zero steps with probability at least $0.99$.
    An application of Stirling's approximation shows that the probability of a random walk that takes $\Omega(\gamma c)$ steps ends at the origin is roughly $O(1/\sqrt{\gamma c})$ which is less than $0.99$ for $c \geq \Omega(1/\gamma)$.
    We thus conclude that $\prob{X_3 = X_0} < 0.1$.
    Next, a calculation shows that $\prob{\max\{X_0, X_3\} = c} = 2\gamma^c < 0.05$ for $c = \Omega(1)$.
    Similarly, $\prob{\min\{X_0, X_3\} = 0} = (1-\gamma)^c < 0.05$ for $c = \Omega(1/\gamma)$.
    In particular, $\prob{X_3 \neq X_0, \max\{X_3, X_0\} \leq c - 1, \min\{X_3, X_0\} \geq 1} \geq 0.8$ by a union bound.
    As discussed above, conditioned on these three events, we lose a trade with probability $1/2$.
    So TR loses a trade, compared to OPT, with probability at least $0.4$ and thus, $\OPT - \TR \geq 0.4$.

    It now suffices to check that $\OPT \leq O(\gamma c)$.
    Note that the following is an optimal matching.
    We first match $\min(X_3, X_0)$ buyers and sellers with values $3$ and $0$, respectively. This contributes $3 \cdot \min(X_3, X_0)$ to the GFT.
    If $X_3 = X_0$ then we do not match any additional agents.
    If $X_3 > X_0$ then we match $X_3 - X_0$ buyers and sellers with value $3$ and $2$, respectively.
    This contributes $X_3 - X_0 = |X_3 - X_0|$ to the GFT.
    If $X_3 < X_0$ then match $X_0 - X_3$ buyers and sellers with value $1$ and $0$, respectively.
    This contributes $X_0 - X_3 = |X_3 - X_0|$ to the GFT.
    In any case, the optimal GFT is given by $3\min(X_3, X_0) + |X_3 - X_0| = \min(X_3, X_0) + X_3 + X_0 \leq 2(X_3 + X_0)$.
    Taking expectations, gives that $\OPT \leq O(\gamma c)$ (since $\expect{X_3} = \expect{X_0} = \gamma c$).
\end{proof}

\subsection{Other Missing Proofs From \Cref{sec:fsd_noniid_approx}}
\label{app:other_proofs}
\begin{proof}[Proof of Lemma~\ref{lemma:opt_monotone}]
    Consider the following coupling procedure to generate a uniform random set of quantiles $\Bq_B', \Bq_S'$ (resp.~$\Bq_B'', \Bq_S''$) conditioned on $B_+ = B', S_+ = S'$ (resp.~$B_+ = B'', S_+ = S''$).
    For $i \in B'$ we sample a uniform random quantile $\qb(i)$ subject to $\qb(i) \geq 1-\gamma$ and set $\qb'(i) = \qb''(i) = \qb(i)$.
    For $i \in [c] \setminus B''$, we sample a uniform random quantile $\qb(i)$ subject to $\qb(i) < 1-\gamma$ and set $\qb'(i) = \qb''(i) = \qb(i)$.
    For $i \in B'' \setminus B'$, we sample a uniform random quantile $\qb''(i)$ conditioned on $\qb''(i) \geq 1-\gamma$ a uniform random quantile $\qb'(i) < 1-\gamma$.
    We use a similar sampling strategy for sellers in $S$.
    Recall that the remainder of the market is independent of the agents in $B$ and $S$.
    Thus, we sample the remainder of the market and assign the same set of quantiles to $\qb(i)$ and $\qs(j)$ for $i \in [m] \setminus B$ and $j \in [n] \setminus S$.

    Let $M'$ (resp.~$M''$) be the market consisting of buyers and sellers with quantiles $(\Bq_B', \Bq_S')$ (resp.~$(\Bq_B'', \Bq_S'')$).
    Observe that our choice of coupling means that for every buyer $i \in [m+c]$, we have $\qb'(i) \leq \qb''(i)$ and for every seller $j \in [n+c]$ we have $\qs'(j) \geq \qs''(j)$.
    Thus any matching in $M'$ is a valid matching in $M''$ where the GFT of the latter is lower bounded by the GFT of the former.
\end{proof}
\section{Lower Bound}\label{subsec:non_fsd_lb}
We derive a lower bound on the number of additional agents with result of \cite{BabaioffGG20}.
\begin{lemma}[{\cite[Proposition~E.5]{BabaioffGG20}}]
\label{lem:bgg_lb}
For any $\eps>0$ and any integer $m, n$, there exists a double auction instance with $m$ buyers, $n$ sellers, and distributions $F_B$, $F_S$ such that: (1) $\Pr_{b\sim F_B, s\sim F_S}[b\geq s]\geq \frac{\eps^2}{(m+n+2c)^8}$; (2) For any deterministic, prior-independent, BIC, IR, BB and anonymous\footnote{A mechanism is \emph{anonymous} if the mechanism treats each buyer and seller equally. In other words, swapping the identity of any two buyers (or sellers) will only result in their allocation and payment being swapped.} mechanism $M$, the GFT of $M$ on $m+c$ buyers and $n+c$ sellers is smaller than $\eps$ times the first-best GFT on $m$ buyers and $n$ sellers.
\end{lemma}

\begin{corollary}\label{cor:non_fsd_lb}
For any sufficiently small $r\in (0,1)$, there exists a double auction instance with \\
$\Pr_{b\sim F_B, s\sim F_S}[b\geq s]=r$ such that: For any $c\leq \frac{1}{4}\cdot r^{-1/8}$ and any deterministic, prior-independent, BIC, IR, BB and anonymous mechanism $M$, the GFT of $M$ on $m+c$ buyers and $n+c$ sellers is smaller than the first-best GFT on $m$ buyers and $n$ sellers.
\end{corollary}
\begin{proof}
Applying \Cref{lem:bgg_lb} with any $n\leq m\leq \frac{1}{4}\cdot r^{-1/8}$, we have $\eps\leq \sqrt{r(m+n+2c)^8}\leq 1$. The proof is done by \Cref{lem:bgg_lb}.  
\end{proof}

\end{document}